%% file: main.tex
\documentclass[sigconf, nonacm]{acmart}
\usepackage{booktabs} 
\usepackage[english]{babel} 
\usepackage[T1]{fontenc} 
\usepackage[utf8]{inputenc} 

\usepackage{amsmath,amsfonts,amsthm,amssymb,bm}
\usepackage[ruled,vlined,linesnumbered]{algorithm2e}
\usepackage{multirow}
\usepackage{stmaryrd}
\usepackage[caption=false,font=footnotesize,hangindent=10pt]{subfig}
\usepackage{array}
\usepackage{enumitem}
\usepackage{algorithmic}
\usepackage[capitalise]{cleveref}
\usepackage{url}
\usepackage[table]{xcolor}
\usepackage{graphicx}
\usepackage{pifont}
\usepackage{wasysym}

\newcommand\vldbdoi{XX.XX/XXX.XX}
\newcommand\vldbpages{XXX-XXX}
\newcommand\vldbvolume{19}
\newcommand\vldbissue{2}
\newcommand\vldbyear{2020}
\newcommand\vldbauthors{\authors}
\newcommand\vldbtitle{\shorttitle} 
\newcommand\vldbavailabilityurl{URL_TO_YOUR_ARTIFACTS}
\newcommand\vldbpagestyle{plain}

\newcommand{\bullethdr}[1]{\noindent\textbullet\ \,\textbf{#1}}

\begin{document}
\title{Efficient and Secure Range Counting over Distributed Geographic Data with Query Range Protection}

\author{Haoxin Yang$^1$, Pinghui Wang$^{1,\ast}$, Zhe Hou$^2$, Tian Zhou$^1$, Guangmingzi Yang$^2$, Zehua Lei$^2$, Rundong Li$^1$, Yutong Song$^1$, Yongyuan Peng$^1$, Fangming Dong$^1$, Xiaohong Guan$^1$}
\thanks{$^\ast$ Pinghui Wang is the corresponding author.}
\affiliation{%
  \institution{$^1$Xi'an Jiaotong University, $^2$China Mobile System Integration Co., Ltd}
}
\email{{yhxaxx1, xjtulirundong, Sunny52012, butter0210field, dfming}@stu.xjtu.edu.cn, phwang@mail.xjtu.edu.cn}
\email{hz4980@163.com, {tianzhou, xhguan}@xjtu.edu.cn, yangguangmingzi@126.com, leizehua@cmict.chinamobile.com}

\begin{abstract}

\input{abstract}

\end{abstract}

\maketitle

\pagestyle{\vldbpagestyle}
\begingroup\small\noindent\raggedright\textbf{PVLDB Reference Format:}\\
\vldbauthors. \vldbtitle. PVLDB, \vldbvolume(\vldbissue): \vldbpages, \vldbyear.\\
\href{https://doi.org/\vldbdoi}{doi:\vldbdoi}
\endgroup
\begingroup
\renewcommand\thefootnote{}\footnote{\noindent
This work is licensed under the Creative Commons BY-NC-ND 4.0 International License. Visit \url{https://creativecommons.org/licenses/by-nc-nd/4.0/} to view a copy of this license. For any use beyond those covered by this license, obtain permission by emailing \href{mailto:info@vldb.org}{info@vldb.org}. Copyright is held by the owner/author(s). Publication rights licensed to the VLDB Endowment. \\
\raggedright Proceedings of the VLDB Endowment, Vol. \vldbvolume, No. \vldbissue\ %
ISSN 2150-8097. \\
\href{https://doi.org/\vldbdoi}{doi:\vldbdoi} \\
}\addtocounter{footnote}{-1}\endgroup

\ifdefempty{\vldbavailabilityurl}{}{
\vspace{.3cm}
\begingroup\small\noindent\raggedright\textbf{PVLDB Artifact Availability:}\\
The source code, data, and/or other artifacts have been made available at \url{https://github.com/jackson-maybe/PPRC}.
\endgroup
}

\input{introduction}

\input{problem}

\input{preliminary}

\input{method}

\input{results}

\input{related}

\input{conclusions}

\balance
\bibliographystyle{ACM-Reference-Format}
\bibliography{references}

\end{document}

%% file: abstract.tex
Range counting is a core primitive in geographic information systems. 
When data is distributed across multiple organizations, conducting range counting raises substantial privacy concerns.
Existing privacy-preserving protocols focus on protecting organizations’ datasets, but cannot simultaneously achieve efficiency, query privacy, and accuracy on overlapping data. 
Typical protocols process query range in plaintext for efficient point-in-range evaluation, since query-private designs rely on expensive secure comparisons. Moreover, most works assume non-overlapping datasets across organizations, which leads to huge errors in overlapping scenarios. 

In this paper, we propose \emph{PPRC}, the first protocol that jointly satisfies all the privacy, efficiency, and accuracy requirements.
PPRC makes two key technical contributions.
First, we design the \emph{Private Range Predicate (PRP)} technique that supports efficient point-in-range evaluation while protecting the query range.
PRP reformulates range evaluation as encrypted membership tests, effectively replacing costly secure comparisons with faster secure multiplications.
Second, we propose \emph{Oblivious Linear Counting (OLC)}, an aggregation scheme that efficiently and securely aggregates partial results from organizations with overlapping data.
OLC involves only lightweight cryptographic operations and ensures that no information is leaked beyond the final range count. 
We theoretically analyze the accuracy, efficiency, and security of PPRC. Experiments on real-world and synthetic datasets show that PPRC achieves up to $55\times$ smaller errors and $37\times$ speedup compared to baseline protocols.

%% file: introduction.tex
\section{Introduction} \label{sec:introduction}

In geographic information systems (GIS), range counting is a fundamental task that returns the number of distinct geographic records falling within a given query range. It supports a variety of geographic analysis and decision-making tasks.

With the growing scale of data and the prevalence of multi-organizational data silos, range counting increasingly needs to be performed over \emph{distributed geographic datasets} owned by multiple data holders.
This gives rise to the \emph{Distributed Range Counting (DRC)} problem, which computes the number of records whose associated locations fall within a query range over the union of all datasets.

Such scenarios naturally arise in applications where geographic data is held by multiple organizations with their own infrastructures.
For example, a car-hailing user queries the number of available cars located within a pickup region to estimate the expected waiting time, while car location records are distributed across different car-hailing companies~\cite{li2023efficient, tong2025hu}.
Similarly, a resident estimates the number of infection cases located within a geographic region using data held by several hospitals to assess local infection risk~\cite{zwiers2024federated}.
In addition, an urban planner analyzes population or traffic density by counting records located within queried geographic regions using data from multiple agencies~\cite{chen2025u, sun2021}.
In these scenarios, the query user only requires an aggregated range count, while both the data holders' datasets and the query range are highly sensitive.

Consequently, such applications involve \emph{two privacy properties}:
\emph{data privacy}, ensuring that each data holder’s records remain confidential,
and \emph{query privacy}, preventing disclosure of the user’s query range, which may reveal personal information such as location or health status.
For instance, a query range centered around a user’s home exposes their residential area, while repeated queries near a hospital can imply potential health issues.
These properties are mandated by major data protection regulations~\cite{voigt2017eu,california2019,chen2021understanding}.

Existing DRC schemes fail to meet the two properties simultaneously.
To achieve high efficiency, most Data-Private works~\cite{li2023efficient, tong2025hu, chen2025u} only protect \textit{data privacy}, while leaving the query range in plaintext. 
In these designs, the query user submits the plaintext query range to each data holder, who computes a partial range count—the number of records in its dataset that fall within the range. Data holders then protect their partial counts via encryption or differential privacy noise before a central aggregator sums them to derive the final range count.
While efficient, this design exposes the query range and thus violates query privacy.

Protecting \textit{query privacy} typically incurs substantial computational overhead.
Existing Query-Private protocols~\cite{akhavan2023level,yu2021psafety,ZhangLZZGWSL24} and Exact multiparty computation (MPC) solutions~\cite{dauterman2022waldo,faisal2023tva} achieve query privacy using fully homomorphic encryption or secret sharing. 
These approaches conduct point-in-range evaluation through secure comparisons between record coordinates and encrypted range boundaries.
However, secure comparisons require expensive cryptographic operations, such as bit decomposition and polynomial evaluation, leading to prohibitive computation costs at scale.

Moreover, most existing DRC protocols assume that datasets across data holders contain \textit{no overlapping records}, which is frequently violated in real-world scenarios—for example, in a car-hailing scenario, a driver can be registered with multiple car-hailing companies. 
Thus, simply summing partial counts—as done in both Data-Private and Query-Private protocols—leads to large overestimation errors.
On the other hand, while Exact MPC solutions \cite{dauterman2022waldo,faisal2023tva} support privacy-preserving deduplication, they rely on costly operations such as oblivious sorting and secure comparisons, making them impractical at scale.

\begin{figure}[t]
  \centering
  \setlength{\abovecaptionskip}{0.1cm}
  \includegraphics[width=\linewidth]{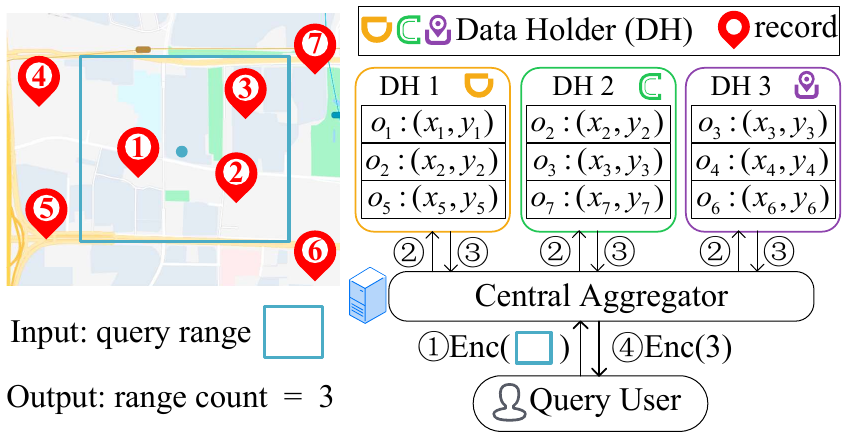}
  \caption{An example of private distributed range counting.}
  \label{fig:problem_description}
\end{figure}

These limitations motivate a stronger problem formulation, which we call \textit{Private Distributed Range Count} (PDRC), explicitly capturing the key challenges that are not jointly addressed by existing protocols.
PDRC requires simultaneously achieving: (1) \textbf{Accuracy}: which ensures accurate range counting over \emph{distinct} records by correctly deduplicating overlaps across data holders;
(2) \textbf{Bilateral Privacy}: which protects both data privacy and query privacy, preventing any leakage about data holders' datasets or the query range;
(3) \textbf{Efficiency}: which provides high efficiency and scalability for practical deployment on large datasets.
\cref{fig:problem_description} illustrates an example with three data holders containing overlapping records, where the correct range count over distinct records is 3.
Throughout the process, both the query range and the data holders' datasets must remain private.
Table~\ref{tab:comparison} further shows that existing protocols fail to satisfy accuracy, bilateral privacy, and efficiency simultaneously.

\begin{table}[htp]
\centering
\renewcommand{\arraystretch}{1.1}
\caption{Comparison of distributed range counting protocols under the PDRC requirements.
\textit{Accuracy}  indicates the robustness against overcounting errors caused by data overlap.
\textit{Privacy} indicates protection of both data and query privacy.
\textit{Efficiency} reflects scalability to large datasets.
Symbols \CIRCLE and \Circle denote \emph{Yes} and \emph{No}, respectively. For \textbf{Privacy}, \LEFTcircle denotes protection of only data privacy.}
\vspace{-0.3cm}
\label{tab:comparison}
\begin{tabular}{lccc}
\toprule
\textbf{Protocol} 
& \textbf{Accuracy} 
& \textbf{Privacy} 
& \textbf{Efficiency} \\
\midrule
Data-Private \cite{li2023efficient, tong2025hu, chen2025u}   & \Circle & \LEFTcircle & \CIRCLE \\
Query-Private \cite{akhavan2023level,yu2021psafety,ZhangLZZGWSL24}      & \Circle & \CIRCLE & \Circle \\
Exact MPC \cite{dauterman2022waldo,faisal2023tva}           & \CIRCLE & \CIRCLE & \Circle \\
\midrule
\textbf{PPRC}                      & \CIRCLE & \CIRCLE & \CIRCLE \\
\bottomrule
\end{tabular}
\end{table}

To address the PDRC problem, we propose \textit{PPRC}, the first protocol that jointly achieves accurate range counting over overlapping datasets, bilateral privacy, and practical efficiency.
Within PPRC, secure range evaluation and privacy-preserving aggregation are jointly designed to enable efficient query-private processing and lightweight deduplication.
Specifically, PPRC makes two key technical contributions.
First, we design the \textbf{Private Range Predicate (PRP)}, a range-predicate abstraction that enables efficient point-in-range evaluation while protecting the query range.
PRP reformulates range evaluation as secure membership tests via encrypted Bloom filters.
This abstraction replaces costly secure comparisons with cheaper secure multiplications, greatly reducing the overall computational cost.
PRP exposes an accuracy-efficiency tradeoff controlled by the Bloom filter size. Larger Bloom filter sizes improve query accuracy, but incur higher computational overhead. 
In addition, larger query ranges require larger Bloom filter encodings, which also increase the computational cost.
Second, we propose \textbf{Oblivious Linear Counting (OLC)}, a secure aggregation scheme designed to efficiently and securely aggregate PRP results from overlapping datasets.
OLC adapts the classic Linear Counting sketch~\cite{whang1990linear} to the encrypted domain, introducing two key technical optimizations to ensure practical performance and privacy:
(i) an efficiency optimization that restructures the sketch aggregation logic to substitute expensive cryptographic multiplications with lightweight additions, drastically reducing the computation cost; and
(ii) a privacy optimization that ensures no information is leaked beyond the final range count. 
Its accuracy-efficiency tradeoff is controlled by the sketch size: larger sketches improve counting accuracy at a higher computational cost.

Our main contributions are summarized as follows:

\begin{itemize}[leftmargin=3ex]

\item We formalize the \textit{Private Distributed Range Count} (PDRC) problem, which requires efficiency, bilateral privacy, and accuracy over overlapping datasets. We propose \textit{PPRC}, the first protocol that satisfies these requirements simultaneously.

\item We design the \textit{Private Range Predicate (PRP)} technique, a range-predicate abstraction that transforms point-in-range tests into secure membership tests, thereby eliminating expensive secure comparisons.

\item We propose \textit{Oblivious Linear Counting (OLC)}, a secure aggregation scheme tailored for overlapping datasets. 
OLC achieves practical efficiency by substituting expensive cryptographic operations with lightweight additions, and simultaneously ensures no information leakage beyond the final range count.

\item We theoretically analyze the accuracy, efficiency, and security of PPRC. Furthermore, experiments on real-world and synthetic datasets show that PPRC achieves up to $55\times$ smaller errors and $37\times$ speedup compared with baseline protocols.

\end{itemize}

%% file: problem.tex
\section{Problem Formulation}
\label{sec:problem}

\subsection{Distributed Range Count}

We consider $I \in \mathbb{N}_+$ data holders. 
Each data holder $i\in \{1,\ldots, I\}$ holds a geographic dataset $D_i$. 
In these datasets, each record is represented as $o=(\text{\textsf{id}},x,y)$, where $\text{\textsf{id}}$ denotes a unique entity identifier and $(x,y)$ denotes the associated geographic coordinates.
Different data holders may have overlapping records, and we denote $D=\bigcup^I_{i=1} D_i$ as the set of all distinct records.

Our goal is to compute the \emph{range count} $RC$, i.e., the number of distinct records in $D$ whose geographic locations fall within a query range $R$ provided by the query user.
The query range is a rectangle $R=([x_{\mathrm{l}}, x_{\mathrm{r}}], [y_{\mathrm{l}}, y_{\mathrm{r}}])$, where $[x_{\mathrm{l}}, x_{\mathrm{r}}]$ and $[y_{\mathrm{l}}, y_{\mathrm{r}}]$ represent the ranges along the $x$- and $y$-dimensions, respectively.
For a record $o=(\text{\textsf{id}},x,y)$, we define the in-range indicator $\mathbf{1}_R(o) = 1$ if $x \in [x_{\mathrm{l}}, x_{\mathrm{r}}] \land y \in [y_{\mathrm{l}}, y_{\mathrm{r}}]$, and $0$ otherwise.
The range count is then defined as:
$$RC = \sum_{o \in D} \mathbf{1}_R(o).$$

Unlike conventional single-data-holder settings
\cite{akhavan2023level,zheng2021efficient,zhang2022efficient},
our problem focuses on distributed range counting over overlapping datasets,
where secure duplicate elimination across multiple data holders is required.
We focus on rectangular ranges, naturally capturing real-world workloads like car-hailing in grid-structured cities.

\subsection{System Model}  \label{subsec:threatmodel}
Our system comprises three entities: the Data Holders (DHs), the Central Aggregator (CA), and the Query User (QU). 
Each DH holds a private dataset and strictly follows the protocol for computations on its data. 
The CA securely aggregates encrypted information from the DHs and sends the encrypted result to the QU. In practice, the CA can be a third-party server or one of the DHs.
The QU generates the homomorphic encryption keys, sends the encrypted range to DHs via the CA, and receives the encrypted query result from the CA. Furthermore, the QU is the sole entity that obtains the plaintext of the range count.

\subsection{Threat Model}
The QU, CA, and DHs are assumed to be \emph{honest-but-curious}, i.e., they strictly follow the protocol but attempt to infer additional information from protocol execution.
In particular, the QU attempts to infer DHs’ private datasets beyond the query result. The CA and each DH attempt to learn the query range submitted by the QU, the query result, or other DHs’ private datasets.
We assume that the CA may collude with one or more DHs, and that DHs may collude among themselves. 
We also assume that the QU does not collude with the CA or any DH. Such collusion is outside the scope of our threat model because it would allow the QU to obtain intermediate protocol information through the colluding party, potentially enabling inference of additional private information from the DHs' datasets beyond the query result.

\subsection{Design Goal}

In this work, our goal is to propose an efficient and secure protocol to compute the range count across overlapping datasets. Specifically, the following properties should be satisfied:
\begin{itemize}[leftmargin=3ex] 

\item \textbf{Accuracy:} The computed range count should reflect the number of \textit{distinct} records within the query range, eliminating duplicated contributions caused by overlaps among datasets.

\item \textbf{Data Privacy:} The dataset of each DH should remain confidential. Neither the other DHs nor the CA can infer any information about it, and the QU learns no additional information about DHs' datasets beyond the final range count.

\item \textbf{Query Privacy:} Only the QU knows the plaintext query range, while DHs and the CA gain no knowledge about it.

\item \textbf{Efficiency:} The protocol should be highly efficient and scalable for large datasets in real-world settings.

\end{itemize}

%% file: preliminary.tex
\section{Preliminaries}
\label{sec:preliminary}

\input{ss_pre}

%% file: ss_pre.tex
\subsection{Bloom Filter}\label{subsec:bloom_filter}
Bloom filter \cite{geravand2013bloom} is a classic data structure used to 
solve the membership test problem, i.e., 
determining whether an element belongs to a given set.
For a set $U$, the Bloom filter encodes all elements of $U$ into an array of $M$ bits, denoted as $\mathrm{BF}[m], 1 \leq m \leq M$. Initially, all bits in the array are set to 0.
Then, we choose $K$ independent hash functions $\{h_1, \ldots, h_K\}$, each mapping $u \in U$ uniformly to an index in $\{1,\ldots,M\}$. To encode $u \in U$ into BF, we set $\mathrm{BF}[h_k(u)]=1$ for $1 \leq k \leq K$.
After encoding all elements of $U$ into the BF, we can test whether an element $u'$ belongs to $U$ by checking whether $\mathrm{BF}[h_k(u')]=1$ for all $1 \leq k \leq K$.
If any bit is 0, $u'$  is definitely not in $U$. If all bits are 1, $u'$ is in $U$ with a high probability. Notably, a Bloom filter can produce \emph{false positive}. Assuming the cardinality of $U$ is $\mu$, the false positive probability can be computed as: 
$
    f_p = (1-(1-\frac{1}{M})^{K\mu})^K \approx (1-e^{-K\frac{\mu}{M}})^K.
$

\subsection{Fully Homomorphic Encryption}\label{subsec:HE}

Fully Homomorphic Encryption (FHE) is a family of cryptographic schemes that enables computations directly over encrypted data~\cite{BourseMMP18}. 
Typically, an FHE scheme supports the following homomorphic computations: i) Add-\uppercase\expandafter{\romannumeral1}: $\mathrm{E}(m_1)+\mathrm{E}(m_2) \rightarrow \mathrm{E}(m_1+m_2)$;
ii) Add-\uppercase\expandafter{\romannumeral2}: $\mathrm{E}(m_1)+m_2 \rightarrow \mathrm{E}(m_1+m_2)$;
iii) Mul-\uppercase\expandafter{\romannumeral1}: $\mathrm{E}(m_1) \cdot \mathrm{E}(m_2) \rightarrow \mathrm{E}(m_1 \cdot m_2)$; 
iv) Mul-\uppercase\expandafter{\romannumeral2}: $\mathrm{E}(m_1) \cdot m_2 \rightarrow \mathrm{E}(m_1 \cdot m_2)$, where $m_1$ and $m_2$ are plaintexts, and $\mathrm{E}(m_1)$ and $\mathrm{E}(m_2)$ are their corresponding FHE ciphertexts.

In our proposed protocol, we employ the FHE scheme as a cryptographic tool and exploit its addition and multiplication homomorphic properties.
Our protocol is instantiated using the Symmetric Homomorphic Encryption (SHE) scheme proposed by Mahdikhani et al.~\cite{journals/iotj/MahdikhaniLZSG20}, which has been used in several searchable encryption schemes~\cite{zheng2021efficient,zhang2022efficient,tu2023pmrk,li2023eppsq}. To evaluate the impact of the underlying homomorphic encryption scheme, we additionally instantiate our protocol with the representative RLWE-based FHE scheme BFV~\cite{2012somewhat} and report the comparison in our experimental evaluation.

Compared with RLWE-based schemes such as BFV~\cite{2012somewhat} and BGV~\cite{BrakerskiGV14}, the adopted SHE scheme relies only on modular integer arithmetic instead of polynomial-ring operations, resulting in lower computational overhead and more compact ciphertexts.
These properties make it particularly suitable for lightweight protocols.

Despite its simplicity, the SHE scheme is CPA-secure, ensuring that encryptions of plaintexts are computationally indistinguishable.
Furthermore, although symmetric, it can be converted into a public-key setting \cite{zhang2024performance}.
Specifically, two encryptions of zero, denoted as $\mathrm{E}(0)_1$ and $\mathrm{E}(0)_2$, generated with different random values, are published as public parameters. Consequently, any message $m$ can be encrypted by a DH using these public parameters as $\mathrm{E}(m) = (m + r_1 \cdot \mathrm{E}(0)_1 + r_2 \cdot \mathrm{E}(0)_2) \mod \mathcal{N}$, where $r_1, r_2 $ are random integers.

\subsection{Mergeable Count Estimation Sketches}\label{subsec:LC}

Mergeable count estimation sketches~\cite{whang1990linear, flajolet1985probabilistiC, durand2003loglog} are compact probabilistic data structures that estimate the union count, i.e., the number of distinct records in the union of multiple datasets. 
One of the most well-known sketches is the linear counting (LC) sketch~\cite{whang1990linear}, which provides accurate estimates for small union counts. It is suitable for scenarios such as online ride-hailing, where users query the count of nearby cars, and the count is typically small.

The LC sketch comprises $S$ counters, all initialized to 0.
Each data record $o$ is uniformly hashed to one of the counters using a hash function $h$, i.e., $\mathrm{LC}[h(o)]=1$.
Let $S'$ denote the number of zero counters; the estimated count is $RC = -S \ln \frac{S'}{S}$.
To merge multiple sketches, all DHs use the same hash function $h$, ensuring identical records map to the same counters.
The sketches are then merged via the bitwise OR to counters with the same indices.
The union count can be computed through the merged sketch.

%% file: method.tex
\section{PPRC Protocol}
\label{sec:method}
In this section, we present our PPRC protocol for solving the PDRC problem.
We first introduce the basic idea of PPRC in~\cref{sec: basic idea}.
We then describe its two core components, namely \textit{Private Range Predicate (PRP)} and \textit{Oblivious Linear Counting (OLC)}.
The main notations are summarized in~\cref{table: notation}.

\begin{table}[t] 
\centering
\caption{Notations in our protocol}
\vspace{-0.3cm}
\begin{tabular}{ll} 
\toprule
\textbf{Notation} & \textbf{Definition}  \\ \midrule
$R$ & the query range held by the QU  \\
$M$ & the bit length of the Bloom filter  \\
$K$ & the number of hash functions of the Bloom filter \\
$I$ & the number of DHs  \\
$D_i$ & the private dataset held by DH $i$ \\
$D_i^{\ast}$ & the set of records in $D_i$ falling within $R$ \\
$N$ & the number of records in the union of $\{D_i\}_{i=1,\ldots,I}$\\
$S$ & the counter length of the LC sketch \\
$S'$ & the number of counters set to 0 in the LC sketch \\
\bottomrule
\end{tabular} \label{table: notation}
\end{table}

\input{ss_overview}

\input{ss_PPRC}

\input{ss_complexity}

\input{ss_security}

%% file: ss_overview.tex
\subsection{Basic Idea} \label{sec: basic idea}
The goal of PPRC is to enable a query user (QU) to obtain the count of distinct records falling within a rectangular range $R=([x_\ell,x_r], [y_\ell,y_r])$ across multiple data holders (DHs), without leaking $R$ or the DHs’ private datasets, and while scaling to millions of records. Achieving this goal raises two technical challenges.

\textbf{(1) Range evaluation: How can a DH determine, in encrypted form, whether a record $o=(\text{\textsf{id}},x,y)$ lies within the query range $R$ without revealing $R$?} 
The goal is not for the DH to learn the plaintext range-evaluation result, but to compute an encrypted indicator $\mathrm{E}(L)$, where $L=1$ indicates that $o$ lies within $R$ and $L=0$ otherwise.
A direct approach is to encrypt the interval endpoints $[x_\ell, x_r]$ and $[y_\ell, y_r]$ using fully homomorphic encryption (FHE), and then perform secure comparisons \cite{yu2021psafety,zhang2022efficient,ZhangLZZGWSL24} of $x$ with $[x_\ell, x_r]$ and $y$ with $[y_\ell, y_r]$ directly on the encrypted data. However, secure comparisons under FHE involve heavy cryptographic operations such as bit decomposition or polynomial evaluation. 
When applied to millions of records per DH, these operations incur prohibitive overhead, making such solutions impractical in practice.

To overcome this, we propose the \textbf{PRP} technique, which reformulates point-in-range evaluation as secure membership tests and implements them using encrypted Bloom filters. This design replaces costly secure comparisons with homomorphic multiplications, thereby improving efficiency. Concretely, the QU encodes all values in $[x_\ell, x_r]$ into a Bloom filter $\mathrm{BF}_{\mathrm{x}}$ with $K$ hash functions $h_{\mathrm{x}1},\ldots,h_{\mathrm{x}K}$, and similarly encodes $[y_\ell, y_r]$ into $\mathrm{BF}_{\mathrm{y}}$. Each bit in these filters is encrypted individually using FHE, and the encrypted Bloom filters, together with the corresponding hash functions, are sent to the DHs via the central aggregator (CA). 
For a record $o=(\text{\textsf{id}},x,y)$, a DH computes the encrypted $x$-membership indicator by homomorphically multiplying the $K$ ciphertexts corresponding to indices $h_{\mathrm{x}1}(x),\dots,h_{\mathrm{x}K}(x)$:
\[
\mathrm{E}(L_{\mathrm{x}}) \;=\; \prod_{k=1}^K \mathrm{E}\big(\mathrm{BF}_{\mathrm{x}}[h_{\mathrm{x}k}(x)]\big),
\]
and analogously obtains $\mathrm{E}(L_{\mathrm{y}})$ for the $y$-coordinate. The encrypted in-range indicator for ``$o$ lies in the query range'' is then
$$\mathrm{E}(L) \;=\; \mathrm{E}(L_{\mathrm{x}})\cdot \mathrm{E}(L_{\mathrm{y}}),$$
which performs the logical AND of the two membership tests.  
PRP eliminates costly secure comparisons and replaces them with only $O(K)$ homomorphic multiplications per record, while revealing no information about the query range.

\textbf{(2) Deduplicated aggregation: How to compute the range count across overlapping DHs' datasets without leaking any information about these datasets?}  
Each DH now holds encrypted in-range indicators for its records after PRP. A natural method is for each DH to construct an FHE-encrypted Linear Counting (LC) sketch of length $S$, where every in-range record with $\mathrm{E}(L)=\mathrm{E}(1)$ is hashed to a counter. Each DH sends its encrypted sketch to the CA, and the CA homomorphically merges sketches by simulating bitwise OR across counters, and forwards the aggregated encrypted sketch to the QU for decryption and estimation. 

However, this approach faces two major limitations:
(i) \textit{Efficiency}: In FHE, the OR operation corresponds to ciphertext multiplication, which is computationally expensive and becomes infeasible at large scale;
(ii) \textit{Privacy}: If the QU directly obtains the decrypted sketch along with the LC hash function $h$, it can conduct an inference attack. Specifically, for any record $o$, if $h(o)=s$ and the $s$-th counter is $0$, then $o$ does not appear in any DH; otherwise, the QU learns that $o$ exists in at least one DH.

PPRC addresses the limitations by designing the \textbf{OLC} scheme. 
OLC adapts the Linear Counting sketch~\cite{whang1990linear} to the encrypted domain, introducing two key optimizations:
    \textbf{(1) Efficiency optimization.} Instead of costly homomorphic ORs, the CA aggregates LC sketches by element-wise homomorphic addition, a much cheaper operation. The aggregated encrypted sketch can be sent to the QU, who decrypts it, counts the number $S'$ of zero counters, and applies the standard LC estimator to obtain the range count:
\[
RC \;=\; -\,S \ln\!\left(\tfrac{S'}{S}\right).
\]
   \textbf{(2) Privacy optimization.} The above inference attack exploits the fact that counter indices and values directly reveal membership information. To prevent this leakage, the CA applies two lightweight transformations that do not alter the range count result: (i) \emph{counter-index obfuscation}, which permutes the order of encrypted counters so indices cannot be linked to records, and (ii) \emph{non-zero masking}, which multiplies each counter ciphertext by a random non-zero plaintext scalar, preserving zero values while randomizing non-zero magnitudes. These transformations ensure the QU can only learn the final count $RC$, but nothing about which records or DHs contributed to specific counters.

\begin{figure}[t]
  \setlength{\abovecaptionskip}{0.1cm}
  \includegraphics[width=\linewidth]{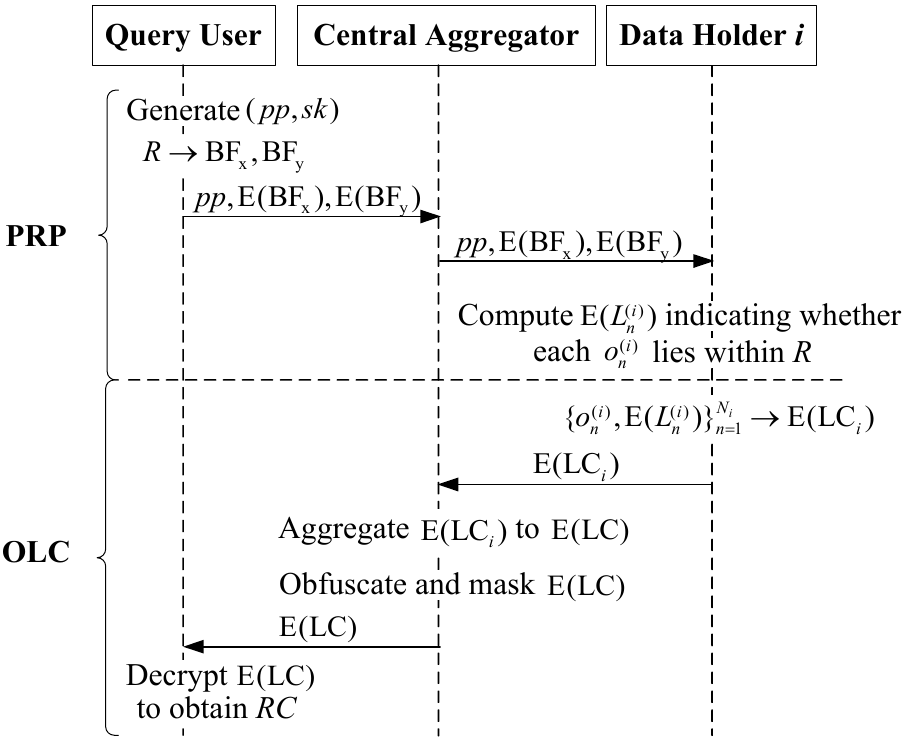}
  \caption{The workflow of the PPRC protocol.}
  \label{fig:protocol_overview}
\end{figure}

In summary, PPRC integrates PRP with OLC to enable efficient, accurate, and privacy-preserving distributed range counting.
PPRC follows a two-stage protocol, as illustrated in \cref{fig:protocol_overview}.
In the first stage, PRP enables each DH to efficiently compute encrypted in-range indicators for its records without learning $R$.
In the second stage, OLC securely aggregates these encrypted results to compute the final range count while correctly handling overlapping records.
We next describe PRP and OLC in detail.

%% file: ss_PPRC.tex
\subsection{Private Range Predicate (PRP)} \label{subsec: PPRC_description}

PRP enables each DH to evaluate whether its records satisfy the query range in encrypted form.
It consists of two steps: \emph{Initialization} and \emph{Execution}.
In the Initialization step, the QU encodes the query range $R$ into encrypted Bloom filters and sends them to DHs.
In the Execution step, each DH performs encrypted range evaluation to compute one encrypted in-range indicator $\mathrm{E}(L_n^{(i)})$ for each record $o_n^{(i)}$, where $L_n^{(i)}=1$ indicates that $o_n^{(i)}$ lies within $R$, and $L_n^{(i)}=0$ otherwise.
These encrypted indicators are then used as the input to OLC for computing the range count.
Next, we detail the two steps.

\begin{algorithm}[t]
	\SetKwInOut{Input}{Input}
	\SetKwInOut{Result}{Result}

	\Input{Query range ($R=[x_\ell, x_r], [y_\ell, y_r])$ held by the QU, where $x_\ell, x_r,y_\ell, y_r \in \mathbb{N}$.}
	\Result{FHE keys $(pp, sk)$, encrypted Bloom filters $\mathrm{E}(\mathrm{BF}_{\mathrm{x}}), \mathrm{E}(\mathrm{BF}_{\mathrm{y}})$, and hash function sets $\mathcal{H}_{\mathrm{x}}, \mathcal{H}_{\mathrm{y}}$.}
    The QU generates the public parameters $pp = (\mathrm{E}(0)_1, \mathrm{E}(0)_2)$, secret key $sk$ for FHE\; \label{seq:set_1}
    The QU generates independent hash functions $\mathcal{H}_{\mathrm{x}} = \{h_{\mathrm{x}1},\dots,h_{\mathrm{x}K}\}$ and $\mathcal{H}_{\mathrm{y}} = \{h_{\mathrm{y}1},\dots,h_{\mathrm{y}K}\}$ for $\mathrm{BF}_{\mathrm{x}}$ and $\mathrm{BF}_{\mathrm{y}}$ constructions, respectively\; \label{seq:set_2_1}
    The QU sets all counters of $\mathrm{BF}_{\mathrm{x}}$ to 0 and $\mathrm{BF}_{\mathrm{y}}$ to 0\;\label{seq:set_2_2}
  \ForEach{$x=x_\ell,x_\ell+1,\ldots,x_r$}{ \label{seq:set_3}
        \For{$k=1, \ldots, K$}{
        $\mathrm{BF}_{\mathrm{x}}[h_{\mathrm{x}k}(x)] = 1$\; \label{seq:set_5}}
}
    \ForEach{$y=y_\ell,y_\ell+1,\ldots,y_r$}{
         \For{$k=1, \ldots, K$}{
        $\mathrm{BF}_{\mathrm{y}}[h_{\mathrm{y}k}(y)] = 1$\;}
}

    \For{$m=1, \ldots, M$}{
Encrypt $\mathrm{BF}_{\mathrm{x}}[m]$ into $\mathrm{E}(\mathrm{BF}_{\mathrm{x}}[m])$\; \label{seq:set_11}
Encrypt $\mathrm{BF}_{\mathrm{y}}[m]$ into $\mathrm{E}(\mathrm{BF}_{\mathrm{y}}[m])$\; \label{seq:set_12}
 
}
     The QU transmits $\mathrm{E}(\mathrm{BF}_{\mathrm{x}}), \mathrm{E}(\mathrm{BF}_{\mathrm{y}}), \mathcal{H}_{\mathrm{x}}, \mathcal{H}_{\mathrm{y}}, pp$ to CA\; \label{line:send}
    CA forwards $\mathrm{E}(\mathrm{BF}_{\mathrm{x}}), \mathrm{E}(\mathrm{BF}_{\mathrm{y}}), \mathcal{H}_{\mathrm{x}}, \mathcal{H}_{\mathrm{y}}, pp$ to $I$ DHs\; \label{line:forward}
	\caption{PRP Initialization}\label{alg:setup_phase}
\end{algorithm}

\noindent \textbf{PRP Initialization.} 
In PRP Initialization, the QU generates encrypted variables and sends them to the CA and DHs.
Specifically, the QU generates the public parameters $pp = (\mathrm{E}(0)_1, \mathrm{E}(0)_2)$ and secret key $sk$ for FHE, two sets of independent hash functions for Bloom filter construction, and two encrypted Bloom filters $\mathrm{E}(\mathrm{BF}_{\mathrm{x}})$, $\mathrm{E}(\mathrm{BF}_{\mathrm{y}})$ representing the query range $R$ in the $x$ and $y$ dimensions. 
The QU then sends $pp$, hash functions, and the encrypted Bloom filters to the CA, which subsequently forwards them to each DH $i=1,\ldots,I$.
The pseudocode is provided in~\cref{alg:setup_phase}.

To construct $\mathrm{E}(\mathrm{BF}_{\mathrm{x}})$ for the $x-$dimension range $[x_\ell,x_r]$, the QU executes the following procedure:

(1) \textit{Setup} (Lines~\ref{seq:set_2_1}-\ref{seq:set_2_2}). All counters $\mathrm{BF}_{\mathrm{x}}[m]$ are set to 0 for $m=1,\ldots,M$, and $K$ independent hash functions $\mathcal{H}_{\mathrm{x}} = \{ h_{\mathrm{x}1},\ldots,h_{\mathrm{x}K} \}$ are generated, each mapping $x$ uniformly to an index in $\{1,\ldots,M \}$.

(2) \textit{Insertion} (Lines~\ref{seq:set_3}-\ref{seq:set_5}).
The QU inserts each $x$ in $\{x_\ell, x_\ell+1, \ldots, x_r\}$ into $\mathrm{BF}_{\mathrm{x}}$.
Concretely, for each $x$, the QU sets $\mathrm{BF}_{\mathrm{x}}[h_{\mathrm{x}k}(x)]=1$ for all $k=1,\ldots,K$.

(3) \textit{Encryption} (Line~\ref{seq:set_11}). Each counter $\mathrm{BF}_{\mathrm{x}}[m]$ is encrypted into $\mathrm{E}(\mathrm{BF}_{\mathrm{x}}[m])$ using $sk$, for all $m=1,\ldots,M$.

The construction of $\mathrm{E}(\mathrm{BF}_{\mathrm{y}})$ follows the same procedure as $\mathrm{E}(\mathrm{BF}_{\mathrm{x}})$.
Together, $\mathrm{BF}_{\mathrm{x}}$ and $\mathrm{BF}_{\mathrm{y}}$ encode the query range in the two coordinate dimensions.
Accordingly, the Bloom filter length $M$ is determined by the encoded range size and the target false-positive rate, and is independent of the DHs' dataset size.

\noindent \textbf{PRP Execution.}
In PRP Execution, each DH $i$ performs range evaluation using the encrypted Bloom filter representation generated from the query range.
For each record $(\text{\textsf{id}}_n^{(i)}, x_n^{(i)}, y_n^{(i)}) \in D_i$, DH $i$ computes an encrypted in-range indicator $\mathrm{E}(L_n^{(i)})$, where $L_n^{(i)}=1$ if the record lies within the query range, and $L_n^{(i)}=0$ otherwise.
This is achieved by reformulating point-in-range testing as secure membership evaluation over the transformed representation.
The procedure is summarized in~\cref{alg:computation_phase}.

\begin{algorithm}[t]
	\SetKwInOut{Input}{Input}
	\SetKwInOut{Result}{Result}

	\Input{Encrypted Bloom filters $\{\mathrm{E}(\mathrm{BF}_{\mathrm{x}}), \mathrm{E}(\mathrm{BF}_{\mathrm{y}})\}$ and the records $\{(x_n^{(i)}, y_n^{(i)})\}_{n=1}^{N_i}$ held by DH $i$.}
	\Result{Encrypted in-range indicators $\{\mathrm{E}(L_n^{(i)})\}_{n=1}^{N_i}$ indicating whether each record lies within $R$.}

 \tcc{Each DH $i$ executes locally:}

 \ForEach{$n=1,\ldots, N_i$}{ \label{seq:check_4}
    $\mathrm{E}(L_{\mathrm{x}n}^{(i)})=\prod_{k=1}^{K} \mathrm{E}(\mathrm{BF}_{\mathrm{x}}[h_{\mathrm{x}k}(x_n^{(i)})])$\;
    $\mathrm{E}(L_{\mathrm{y}n}^{(i)})=\prod_{k=1}^{K} \mathrm{E}(\mathrm{BF}_{\mathrm{y}}[h_{\mathrm{y}k}(y_n^{(i)})])$\;
    $\mathrm{E}(L_n^{(i)})=\mathrm{E}(L_{\mathrm{x}n}^{(i)}) \cdot \mathrm{E}(L_{\mathrm{y}n}^{(i)})$\; \label{seq:check_7}
 }
	\caption{PRP Execution}\label{alg:computation_phase}
\end{algorithm}

Concretely, we exploit homomorphic multiplications to evaluate Bloom filter membership tests.
For each $x_n^{(i)}$, DH $i$ computes: 
$$\mathrm{E}(L_{\mathrm{x}n}^{(i)})=\prod_{k=1}^{K} \mathrm{E}(\mathrm{BF}_{\mathrm{x}}[h_{\mathrm{x}k}(x_n^{(i)})]),$$ 
where $L_{\mathrm{x}n}^{(i)} = 1$ if $x_n^{(i)} \in [x_\ell, \mathrm{x_{r}}]$ and $0$ otherwise.
The intuition is as follows. If $x_n^{(i)} \in [x_\ell, x_r]$, then for each $k=1,\ldots,K$, we have $\mathrm{BF}_{\mathrm{x}}[h_{\mathrm{x}k}(x_n^{(i)})]=1$, and thus $L_{\mathrm{x}n}^{(i)} = 1$. Conversely, if $x_n^{(i)} \notin [x_\ell, x_r]$, then at least one of these positions is $0$, leading to $L_{\mathrm{x}n}^{(i)} = 0$.

The computation for $\mathrm{E}(L_{\mathrm{y}n}^{(i)})$, indicating whether $y_n^{(i)} \in [y_\ell, y_r]$, follows an analogous process, where $L_{\mathrm{y}n}^{(i)} = 1$ indicates $y_n^{(i)} \in [y_\ell, y_r]$ and $L_{\mathrm{y}n}^{(i)} = 0$ otherwise.

Since $(x_n^{(i)}, y_n^{(i)})$ lies within $R$ if and only if $(x_n^{(i)} \in [x_\ell, x_r]) \land (y_n^{(i)} \in [y_\ell, y_r])$, the final encrypted in-range indicator $\mathrm{E}(L_n^{(i)})$ is then obtained as: 
$\mathrm{E}(L_n^{(i)})=\mathrm{E}(L_{\mathrm{x}n}^{(i)}) \cdot \mathrm{E}(L_{\mathrm{y}n}^{(i)}).$

Compared with secure comparisons under FHE, this design transforms the point-in-range evaluation into a lightweight sequence of encrypted Bloom filter lookups. The resulting encrypted in-range indicators require only $O(K)$ homomorphic multiplications per record, making the protocol scalable in practice.

\noindent \textit{Running Example.}
\cref{fig:PPRC1} illustrates the PRP with a simple example for the $x$-coordinate.
Suppose QU encodes the range $\{x_1, x_2\}$ into an encrypted Bloom filter $\mathrm{E}(\mathrm{BF}_{\mathrm{x}})$ and sends it to DHs.
To determine whether $x_1^{(i)}$ lies within the range, DH $i$ computes the hash positions $h_{\mathrm{x}1}(x_1^{(i)})$ and $h_{\mathrm{x}2}(x_1^{(i)})$, 
retrieves the corresponding encrypted bits $\{ \mathrm{E}(1), \mathrm{E}(1)\}$ from $\mathrm{E}(\mathrm{BF}_{\mathrm{x}})$, and performs homomorphic multiplication:
$\mathrm{E}(L_{\mathrm{x}1}^{(i)}) = \mathrm{E}(1) \times \mathrm{E}(1) = \mathrm{E}(1)$,
indicating that $x_1^{(i)}$ is within the range.
For another record $x_2^{(i)}$, one of the retrieved bits equals $\mathrm{E}(0)$, yielding 
$
\mathrm{E}(L_{\mathrm{x}2}^{(i)})  = \mathrm{E}(0),
$
which shows that $x_2^{(i)}$ is outside the range.
The same procedure is then performed on the $y$-dimension, and a record is considered within the query range only if both dimensions satisfy their corresponding conditions.

\begin{figure}[tp]
  \setlength{\abovecaptionskip}{0.1cm}
  \includegraphics[width=\linewidth]{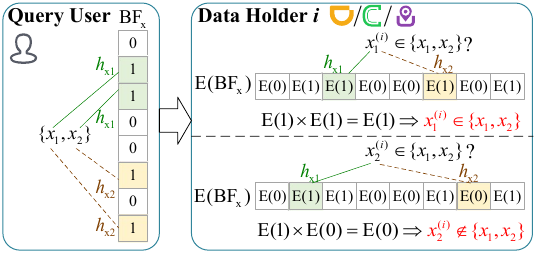}
  \caption{ Example for checking whether $x_n^{(i)}$ lies within $ \{x_1,x_2\}$ without revealing the query range. The length of the Bloom filter is $M=8$, and the number of hash functions is $K=2$. For clarity, only the procedure for the $x$-coordinate is illustrated; the $y$-coordinate is processed analogously.}
  \label{fig:PPRC1}
\end{figure}

\subsection{Oblivious Linear Counting (OLC)} \label{subsec: PPRC_OLC}
During OLC, each DH locally constructs an encrypted LC sketch based on its records and PRP-generated encrypted indicators, the CA securely aggregates the encrypted sketches from all $I$ DHs, and the QU computes the final range count $RC$ from the aggregated result.
The pseudocode of OLC is given in~\cref{alg:pprc_aggregation},  which consists of three steps: \textit{Local Encoding}, \textit{Secure Aggregation}, and \textit{Decryption \& Estimation}.

(1) \textit{Local Encoding} (Lines~\ref{seq:agg_pprc_plus_1}-\ref{seq:agg_pprc_plus_5}). 
Each DH $i$ initializes its sketch $\mathrm{LC}_i$ with all counters set to 0, and encrypts it into $\mathrm{E}(\mathrm{LC}_i)$ using the public parameters $pp$.
Besides, all DHs agree on a global public hash function $h$ that maps each record uniformly into $\{1,\ldots,S\}$, where $S$ is the counter length of the sketch.
The length $S$ is a tunable parameter determined by the target estimation accuracy and the expected range count.
Each DH $i$ then encodes its records into the encrypted sketch by leveraging the encrypted in-range indicators $\mathrm{E}(L_n^{(i)})$ derived from PRP, where $L_n^{(i)}=1$ if record $o_n^{(i)}$ lies within range $R$ and $0$ otherwise.
Specifically, the sketch is updated as:
$$\mathrm{E}(\mathrm{LC}_i[h(o_n^{(i)})]) = \mathrm{E}(\mathrm{LC}_i[h(o_n^{(i)})]) + \mathrm{E}(L_n^{(i)}).$$
For each record $o_n^{(i)}$, if $L_n^{(i)}=1$, the corresponding counter is incremented by one; otherwise, it remains unchanged. Therefore, \textbf{only in-range records} contribute to the sketch, enabling the LC sketch to compute the range count.
Multiple in-range records may hash to the same counter, which is expected in Linear Counting and enables compact representation using a fixed-length sketch.

\begin{algorithm}[t]
	\SetKwInOut{Input}{Input}
	\SetKwInOut{Result}{Result}

	\Input{Dataset $\{(o_n^{(i)}, \mathrm{E}(L_n^{(i)}))\}_{n=1}^{N_i}$ held by DH $i$, where $i=1,\ldots,I$. }
	\Result{The QU obtains the range count $RC$.}

    \tcc{Step 1: Local Encoding by DHs}
    All DHs agree on the global public hash function $h:o\rightarrow\{1,\ldots,S\}$ for sketch construction\; \label{seq:agg_pprc_plus_1}
    \ForEach{$\mathrm{DH}$ $i=1,\ldots,I$}{
    \ForEach{$s=1,\ldots,S$}{
    Initialize the sketch counter $\mathrm{LC}_i[s]=0$\;
    Encrypt $\mathrm{LC}_i[s]$ into $\mathrm{E}(\mathrm{LC}_i[s])$ using $pp$;
  }

    \ForEach{$n=1,\ldots,N_i$}{
    $\mathrm{E}(\mathrm{LC}_i[h(o_n^{(i)})]) = \mathrm{E}(\mathrm{LC}_i[h(o_n^{(i)})]) + \mathrm{E}(L_n^{(i)})$\; \label{seq:agg_pprc_plus_5}
 }
 }

 \tcc{Step 2: Secure Aggregation by CA}
 Each DH $i$ sends its $\mathrm{E}(\mathrm{LC}_i)$ to the CA\; \label{seq:agg_pprc_plus_6}
 
 \ForEach{$s=1,\ldots,S$}{
     Aggregate encrypted sketches: $\mathrm{E}(\mathrm{LC}[s]) = \sum_{i=1}^{I} \mathrm{E}(\mathrm{LC}_i[s])$\;
     Generate a plaintext random integer $r$\;
     Compute $\mathrm{E}(\mathrm{LC}[s]) = \mathrm{E}(\mathrm{LC}[s]) \cdot r$\;
 }
 Shuffle encrypted counters in $\mathrm{E}(\mathrm{LC})$\; \label{seq:agg_pprc_plus_11}
 
 \tcc{Step 3: Decryption \& Estimation by QU}
 The CA sends $\mathrm{E}(\mathrm{LC})$ to the QU\; \label{seq:agg_pprc_plus_12}
 The QU decrypts $\mathrm{E}(\mathrm{LC})$ into $\mathrm{LC}$ using $sk$\;
 Count the number of zero-value counters $S'$ in $\mathrm{LC}$\;
 Compute the range count: $RC = -S\log{\frac{S'}{S}}$\; \label{seq:agg_pprc_plus_15}
	\caption{Oblivious Linear Counting (OLC)}\label{alg:pprc_aggregation}
\end{algorithm}

(2) \textit{Secure Aggregation} (Lines~\ref{seq:agg_pprc_plus_6}-\ref{seq:agg_pprc_plus_11}). Each DH $i$ transmits its encrypted sketch $\mathrm{E}(\mathrm{LC}_i)$ to the CA.
The CA then aggregates counters with the same index across all encrypted sketches: 
$\mathrm{E}(\mathrm{LC}[s]) = \sum_{i=1}^{I} \mathrm{E}(\mathrm{LC}_i[s]), s=1,\ldots,S.$

To mitigate information leakage from decrypted counters, the CA applies two lightweight transformations that preserve the standard LC estimator:
(i) \textbf{counter-index obfuscation}, which permutes the order of counters to break linkage between indices and records; and
(ii) \textbf{non-zero masking}, which multiplies each non-zero counter by a random non-zero plaintext integer, preserving zeros while randomizing magnitudes.
These lightweight transformations preserve the number of zero counters and therefore do not affect the LC estimator, which depends solely on the count of zeros.
As a result, the QU can correctly recover the range count while gaining no information about individual record membership or the contributions of specific DHs, providing robustness against inference attacks with minimal computational overhead. 

(3) \textit{Decryption \& Estimation} (Lines~\ref{seq:agg_pprc_plus_12}-\ref{seq:agg_pprc_plus_15}). 
The CA transmits $\mathrm{E}(\mathrm{LC})$ to the QU, who decrypts it into $\mathrm{LC}$ using the secret key $sk$.
The QU then counts the number of zero counters $S'$ in the decrypted sketch.
Following the standard Linear Counting estimator \cite{whang1990linear}, the fraction of zero counters $\frac{S'}{S}$
approximates the probability that a counter remains unoccupied after all in-range records are inserted.
Therefore, the number of inserted records, i.e., the range count, is estimated as
$$RC=-S\ln{\frac{S'}{S}}.$$
This approach enables the QU to obtain the range count without learning any individual record, thus preserving data privacy.

\noindent \textit{Running Example.}
\cref{fig:PPRC_agg} illustrates the workflow of the OLC using an example with sketch length $S=8$. 
Each DH $i$ locally constructs an encrypted sketch $\mathrm{E}(\mathrm{LC}_i)$, which encodes only the in-range records. 
The CA then performs secure aggregation by summing the encrypted counters across DHs, followed by counter-index obfuscation and non-zero masking, ensuring that the QU cannot infer any information about DHs' datasets even after decryption. 
Finally, the QU decrypts the aggregated sketch.
After decryption, the QU observes $S'=3$ zero counters out of $S=8$.
This means that $5$ counters have been occupied by hashed in-range records.
Applying the LC estimator gives
$RC=-8\ln(\frac{3}{8})$,
which is the estimated range count.

\section{PPRC Analysis} \label{sec:analysis}
In this section, we formally analyze the accuracy, efficiency, and security of the PPRC protocol.

\subsection{\textbf{Accuracy Analysis}}
\begin{figure}[t]
  \setlength{\abovecaptionskip}{0.1cm}
  \includegraphics[width=\linewidth]{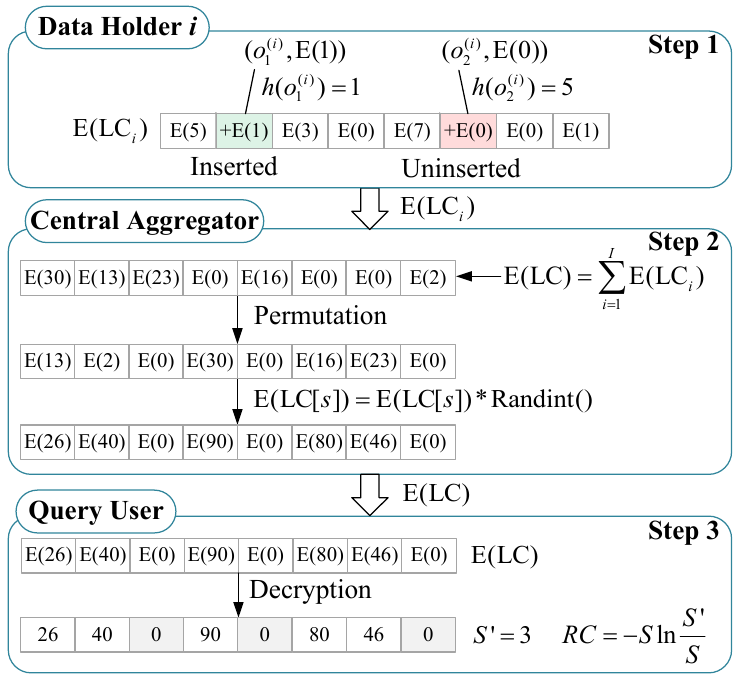}
  \caption{Example for the OLC in the PPRC protocol. Here, the length of LC sketches $S$ is set to 8.}
  
  \label{fig:PPRC_agg}
\end{figure}

\begin{theorem}
Suppose the true range count is $RC$, the estimated value is $\hat{RC}$,
and the false positive rate of the Bloom filters is $f_p$.
Define $\beta=(f_p)^2$, $t=\frac{\beta(N-RC)}{S}$, and
$w=(1-\frac{1}{S})^{RC}$, where $N$ is the number of records in the union of $\{D_i\}_{i=1,\ldots,I}$ and $S$ is the counter length of the LC sketch.
Then the estimator used in PPRC satisfies
\begin{equation}
\mathrm{Bias}(\frac{\hat{RC}}{RC})=\mathbb{E}(\frac{\hat{RC}}{RC}-1)=\frac{2\beta N+w^{-1}e^{-t}-1}{2RC},
\end{equation}

\begin{equation}
\mathrm{StdError}(\frac{\hat{RC}}{RC})= \mathrm{Std}(\frac{\hat{RC}-RC}{RC})=\frac{\sqrt{S(\frac{e^t}{w}-1)}}{RC}.
\end{equation}
\end{theorem}

\begin{proof}[Proof Sketch]
After PRP execution, the number of objects identified within the query range is $RC+r$, where $r$ denotes the number of false positives introduced by the two independent Bloom filters.
Since each filter has a false positive rate $f_p$, we have $\mathbb{E}(r)=(f_{p})^2N=\beta N$.
Besides, $r$ follows a binomial distribution: $r\sim \binom{N-RC}{r}(\beta)^r (1-\beta)^{N-RC-r}$.

During OLC, PPRC uses the LC sketch estimator for $RC+r$ elements.
Let $S'$ denote the number of zero counters in the aggregated sketch.
Based on the standard analysis of the LC sketch, we obtain
$
\mathbb{E}(S')
=
Sw\left(1-\frac{\beta}{S}\right)^{N-RC}
\approx
Swe^{-t},
$
where $w=(1-\frac{1}{S})^{RC}$ and
$t=\frac{\beta(N-RC)}{S}$.
Similarly, we obtain
$
\mathrm{Var}(S') =\mathbb{E}((S')^2)-(\mathbb{E}(S'))^2
\approx
Swe^{-t}-Sw^2e^{-2t}.
$

Let $V=S'/S$ and define $f(V)=-\ln V$.
Applying a second-order Taylor expansion to the LC estimator
$\hat{RC}=Sf(V)$ around $\mathbb{E}(V)$ yields
$
\mathbb{E}(\hat{RC})
\approx
RC+\beta N+\frac{w^{-1}e^t-1}{2}.
$
Therefore, $\mathrm{Bias}(\frac{\hat{RC}}{RC})=\mathbb{E}(\frac{\hat{RC}}{RC}-1)=\frac{2\beta N+w^{-1}e^{-t}-1}{2RC}$.

Using the first-order approximation further yields
$
\mathrm{Var}(\hat{RC})
\approx
\frac{S(e^t/w-1)}{RC^2}.
$
Then, taking the square root yields the standard error: $\mathrm{StdError}(\frac{\hat{RC}}{RC})=\frac{\sqrt{S(\frac{e^t}{w}-1)}}{RC}$.

Due to space limitations, the detailed derivation of the accuracy analysis is available in our public repository.\footnote{\url{https://github.com/jackson-maybe/PPRC}}

\end{proof}

%% file: ss_complexity.tex
\subsection{Protocol Efficiency} \label{protocol_efficiency}

\begin{table}[t]
\caption{Theoretical computation cost, communication cost, and communication rounds of PPRC.}
\vspace{-0.2cm}
\centering
\setlength{\tabcolsep}{2.5pt}
\renewcommand{\arraystretch}{1.2}
\resizebox{1.0\linewidth}{!}{
\begin{tabular}{c|ccc|ccc|c}
\hline
\multirow{2}{*}{\textbf{Component}} &
\multicolumn{3}{c|}{\textbf{\#Ciphertext Mults}} &
\multicolumn{3}{c|}{\textbf{Data Sent}} &
\multirow{2}{*}{\textbf{Rounds}} \\
 & \textbf{QU} & \textbf{DH} & \textbf{CA} &
 \textbf{QU} & \textbf{DH} & \textbf{CA} & \\
\hline
PRP & -- & $O(K\overline{N})$ & -- & $O(M)$ & -- & $O(IM)$ & 2 \\
OLC & -- & -- & -- & -- & $O(S)$ & $O(S)$ & 2 \\
\hline
\textbf{Total} & -- & $O(K\overline{N})$ & -- & $O(M)$ & $O(S)$ & $O(IM+S)$ & 4 \\
\hline
\end{tabular}
}
\label{tab:theoretical_costs}
\end{table}

We theoretically analyze the computation cost, communication cost, and communication rounds of PPRC in Table~\ref{tab:theoretical_costs}.
The analysis is organized by the two protocol components, \textit{PRP} and \textit{OLC}, and further specifies the workloads of the three entities: the QU, a single DH, and the CA.
Since ciphertext--ciphertext multiplications (Mul-\uppercase\expandafter{\romannumeral1} operations defined in \cref{subsec:HE}) dominate the runtime, we use their counts (\#Ciphertext Mults) to measure the computational cost, while omitting other lower-overhead operations from the asymptotic analysis.
Communication cost is measured by the total amount of transmitted data, while communication rounds are measured by the number of message exchanges.

\noindent\textit{PRP.} 
The technique consists of the PRP Initialization and PRP Execution steps.
In PRP Initialization, the QU encrypts the Bloom filter representation of the query range and sends it to the CA, incurring $O(M)$ communication, where $M$ denotes the bit length of Bloom filters.
The CA then forwards the encrypted Bloom filters to the $I$ DHs, resulting in $O(IM)$ data sent.
No ciphertext multiplications are required in this step.

In the PRP Execution, each DH locally determines whether its records fall within the query range using the encrypted Bloom filters.
For each record, the DH requires $O(K)$ ciphertext multiplications, as shown in \cref{alg:computation_phase}.
Hence, the total computational cost is $O(K\overline{N})$ ciphertext multiplications per DH, where $\overline{N}$ denotes the average number of records among the $I$ DHs.
Since this step is performed entirely locally, it incurs no communication cost.
Consequently, considering the communication flow in both steps, PRP requires two communication rounds: QU$\rightarrow$CA$\rightarrow$DH.

\noindent\textit{OLC.} 
After PRP, each DH securely encodes its in-range records into an encrypted LC sketch $\mathrm{E}(\mathrm{LC}_i)$ of size $O(S)$ and sends it to the CA.
The CA then aggregates these encrypted sketches through element-wise addition to obtain the encrypted sketch $\mathrm{E}(\mathrm{LC})$ of size $O(S)$, which is then sent to the QU for decryption.
Benefiting from the optimized aggregation design, OLC requires no ciphertext multiplications and incurs only linear communication cost in $S$.
Following this transmission pattern, OLC also requires two communication rounds: DH$\rightarrow$CA$\rightarrow$QU.

\noindent\textit{Overall.} 
The total computational cost is thus dominated by $O(K\overline{N})$ ciphertext multiplications at the DH, and the total communication cost is $O(IS+IM)$.
Combining the interaction patterns of PRP and OLC, the protocol requires four communication rounds.

%% file: ss_security.tex
\subsection{Protocol Security} \label{protocol_security}

\begin{theorem} \label{theorem:secure_analysis}
When a semi-honest adversary $\mathcal{A}$ corrupts up to $I$ DHs and the CA, PPRC guarantees that $\mathcal{A}$ cannot access any information about the final range count result, the query range, or the private datasets of uncorrupted DHs.
\end{theorem}

\proof
The proof follows the standard simulation-based security definition for semi-honest adversaries.
The definition demonstrates that if a party can simulate the adversary's view without access to any private inputs, the protocol preserves the confidentiality of private inputs without any information leakage.

Let $Sim_{\mathcal{A}}$ denote a simulator that simulates the view of the adversary $\mathcal{A}$ during protocol execution.
During PRP, the adversary $\mathcal{A}$ observes only the encrypted query Bloom filters,
$\mathrm{E}(\mathrm{BF}_{\mathrm{x}})$ and $\mathrm{E}(\mathrm{BF}_{\mathrm{y}})$, and the encrypted indicators computed locally by DHs.
Due to the CPA-security of the FHE scheme~\cite{journals/iotj/MahdikhaniLZSG20}, all these ciphertexts are computationally indistinguishable from encryptions of random values.
Therefore, a simulator $Sim_{\mathcal{A}}$ can efficiently generate an indistinguishable view without access to the underlying plaintexts.
Furthermore, since each encrypted indicator $\mathrm{E}(L_n^{(i)})$ hides its plaintext, $\mathcal{A}$ cannot determine whether it encrypts $0$ or $1$, nor test whether two ciphertexts encrypt the same plaintext.
As a result, $\mathcal{A}$ cannot distinguish which records lie within the query range and which do not, and thus obtain no information about the query range from the encrypted indicators.
In the OLC, each DH $i$ constructs an encrypted LC sketch $\mathrm{E}(\mathrm{LC}_i)$ and the CA aggregates these encrypted sketches. During this process, the adversary $\mathcal{A}$ receives encrypted LC sketches from uncorrupted DHs.
It is easy for $Sim_{\mathcal{A}}$ to generate values computationally indistinguishable from these ciphertexts based on the CPA-security of the SHE scheme.

Thus, the simulator $Sim_{\mathcal{A}}$ can simulate the adversary's view without private inputs, proving PPRC satisfies simulation-based security. Therefore, a semi-honest adversary corrupting up to $I$ DHs and the CA gains no private information.
\endproof

%% file: results.tex
\section{Evaluation} \label{sec:results}

\input{exp_setup}

\input{exp_results}

%% file: exp_setup.tex
\subsection{Experimental Setup}

\noindent \textbf{Datasets}. We evaluate PPRC on three real-world datasets: Yelp~\cite{dataset1}, Brightkite, and Gowalla~\cite{cho2011friendship}, as well as a synthetic dataset for scalability analysis.

\bullethdr{Real-world Datasets.}
Yelp is a business dataset containing merchant locations across multiple regions, and we extract 21,900 geographic records in Florida from Yelp.
Brightkite and Gowalla are location-based social network datasets with 115,383 and 196,561 user check-ins.

\bullethdr{Synthetic Dataset.} 
We generate synthetic locations uniformly within the region of San Francisco (latitude: 37.5--37.9, longitude: $-$122.6--$-$122.2).
The dataset size varies from $10^4$ to $10^7$ records to evaluate scalability.

Following \cite{zhang2022efficient}, we scale locations to integers, where each integer represents a 20m line, and each area corresponds to a 20m $\times$ 20m grid. 
Query ranges are randomly generated as rectangular regions.
We use $2\,\mathrm{km} \times 2\,\mathrm{km}$ ranges as the default setting, and evaluate the impact of different range sizes in \cref{subsec:prange}.
To simulate distributed data settings, each of the $I$ DHs receives a randomly sampled fraction $P$ of the global dataset as its private dataset.
We set $P \in [0.05,0.1,0.15,0.2]$ and $I \in [5,10,15,20]$.

Due to space limitations, additional results and implementation details are available at our public repository.

\vspace{0.1cm}
\noindent \textbf{Metrics}. Our experiments evaluate both accuracy and efficiency.

\begin{figure*}[ht]
  \centering
  \setlength{\abovecaptionskip}{0.3cm}
  \includegraphics[width=1.0\linewidth]{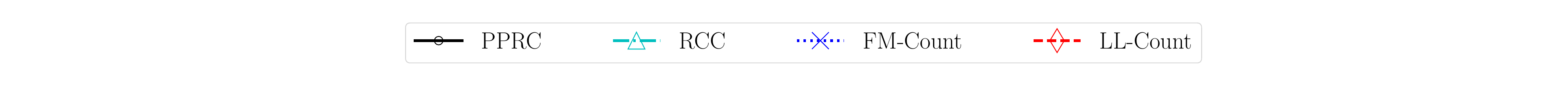}\\
  \vspace{-0.53cm}
  \subfloat[MAE vs. $I$, Yelp]{%
    \includegraphics[width=.25\linewidth]{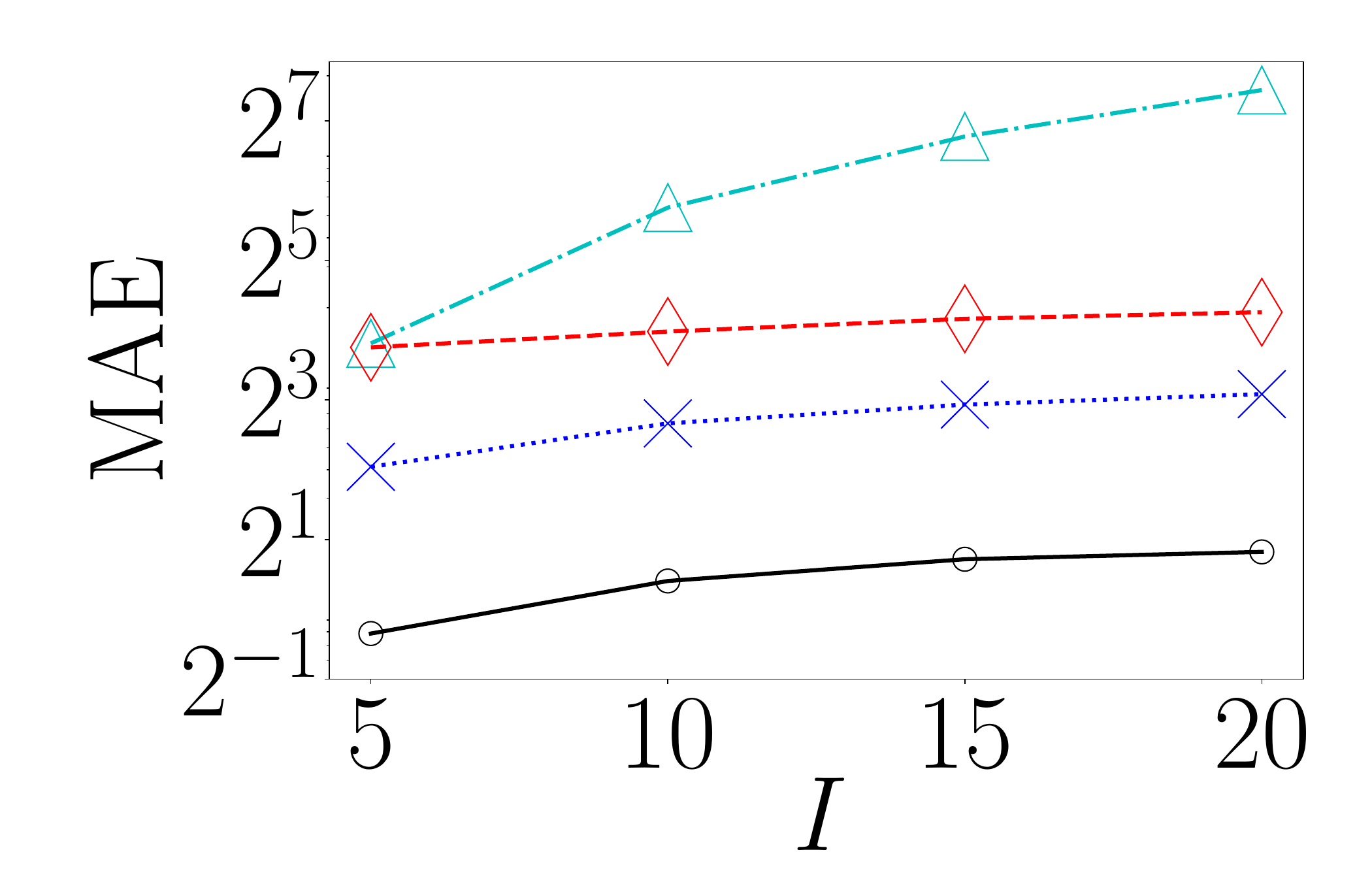}}
  \subfloat[MAE vs. $I$, Brightkite]{%
    \includegraphics[width=.25\linewidth]{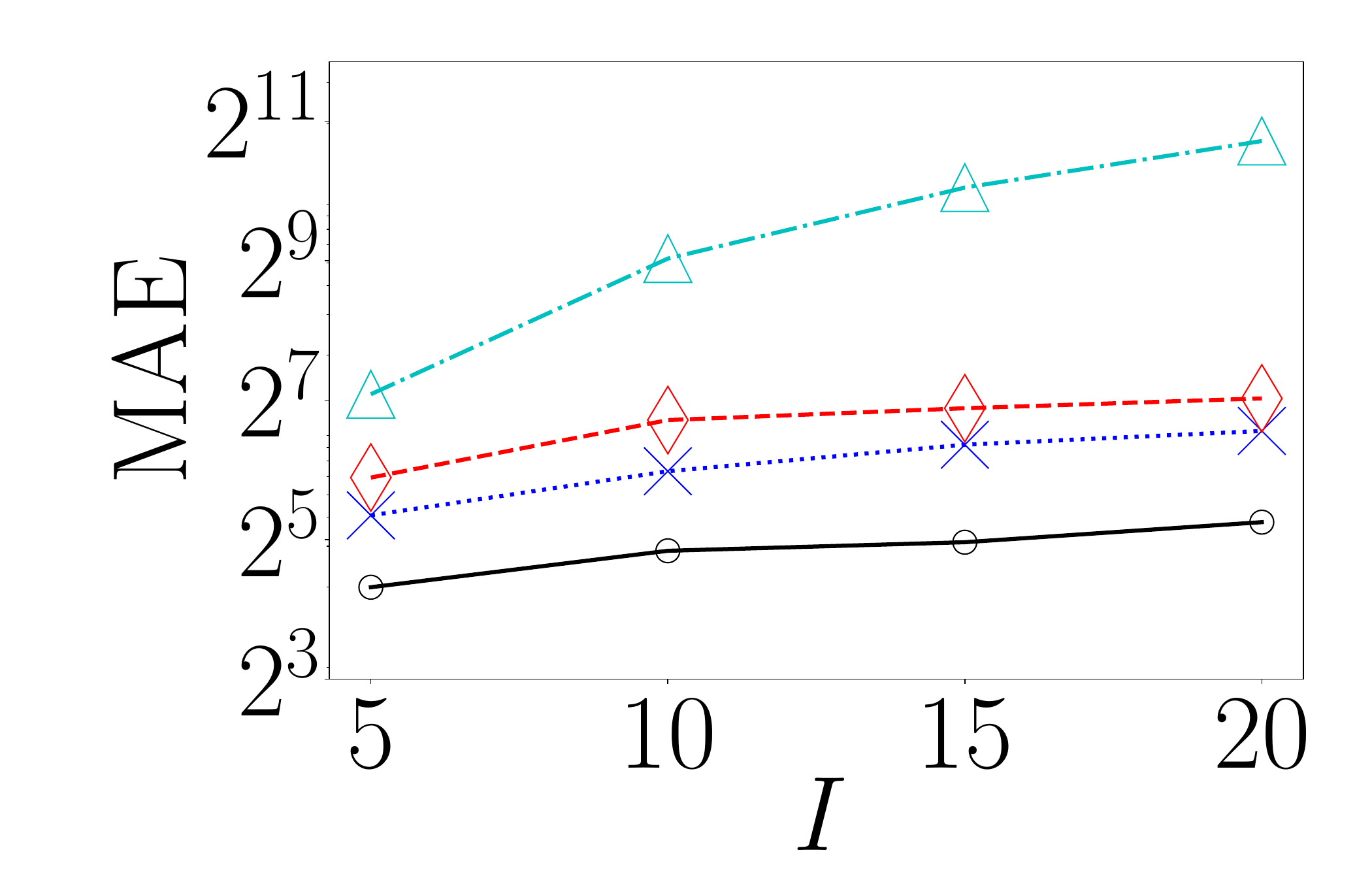}}
  \subfloat[MRE vs. $I$, Yelp]{%
    \includegraphics[width=.25\linewidth]{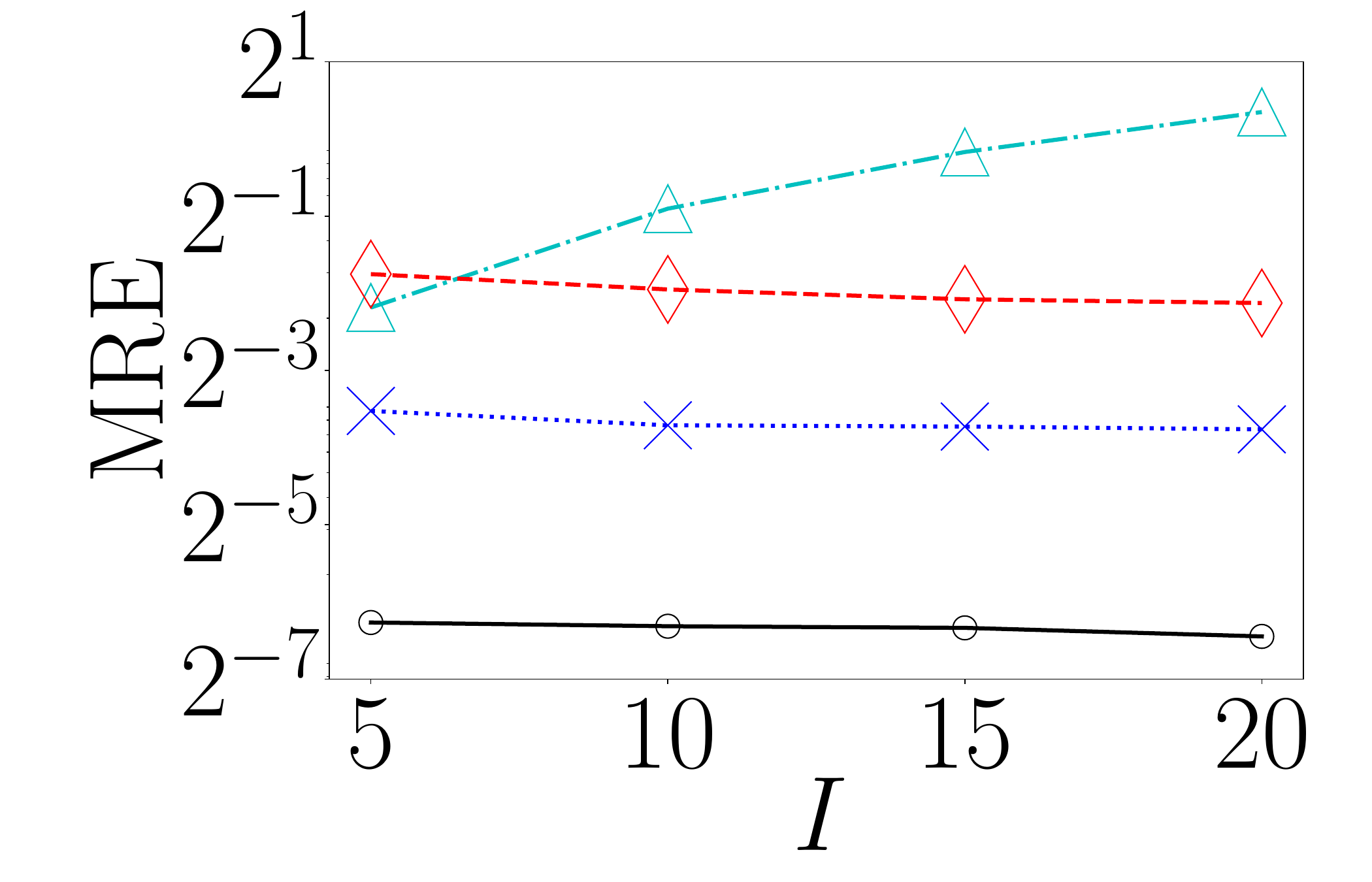}}
  \subfloat[MRE vs. $I$, Brightkite]{%
    \includegraphics[width=.25\linewidth]{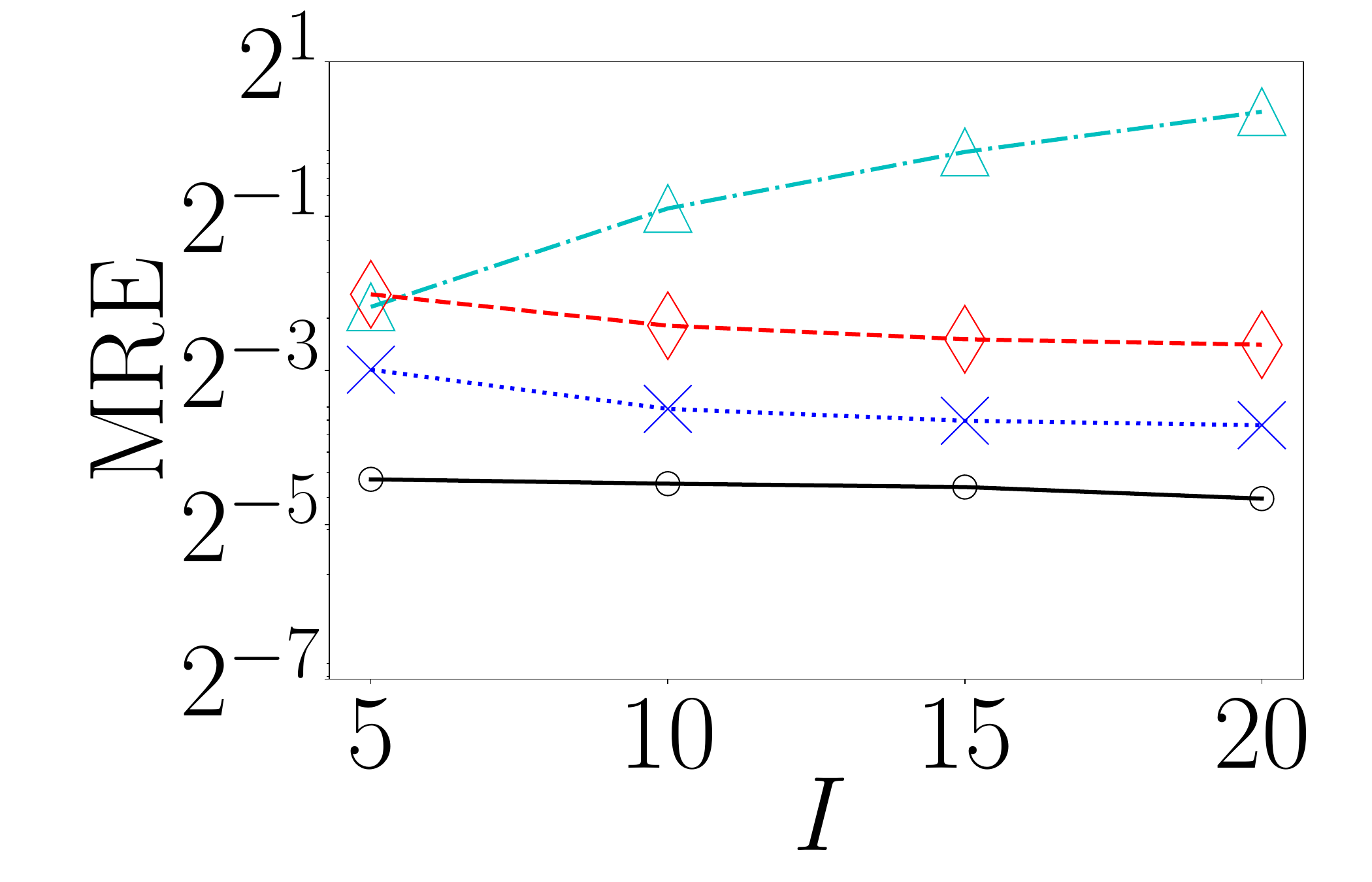}}\\
    \vspace{-0.15cm}
  \subfloat[MAE vs. $P$, Yelp]{%
    \includegraphics[width=.25\linewidth]{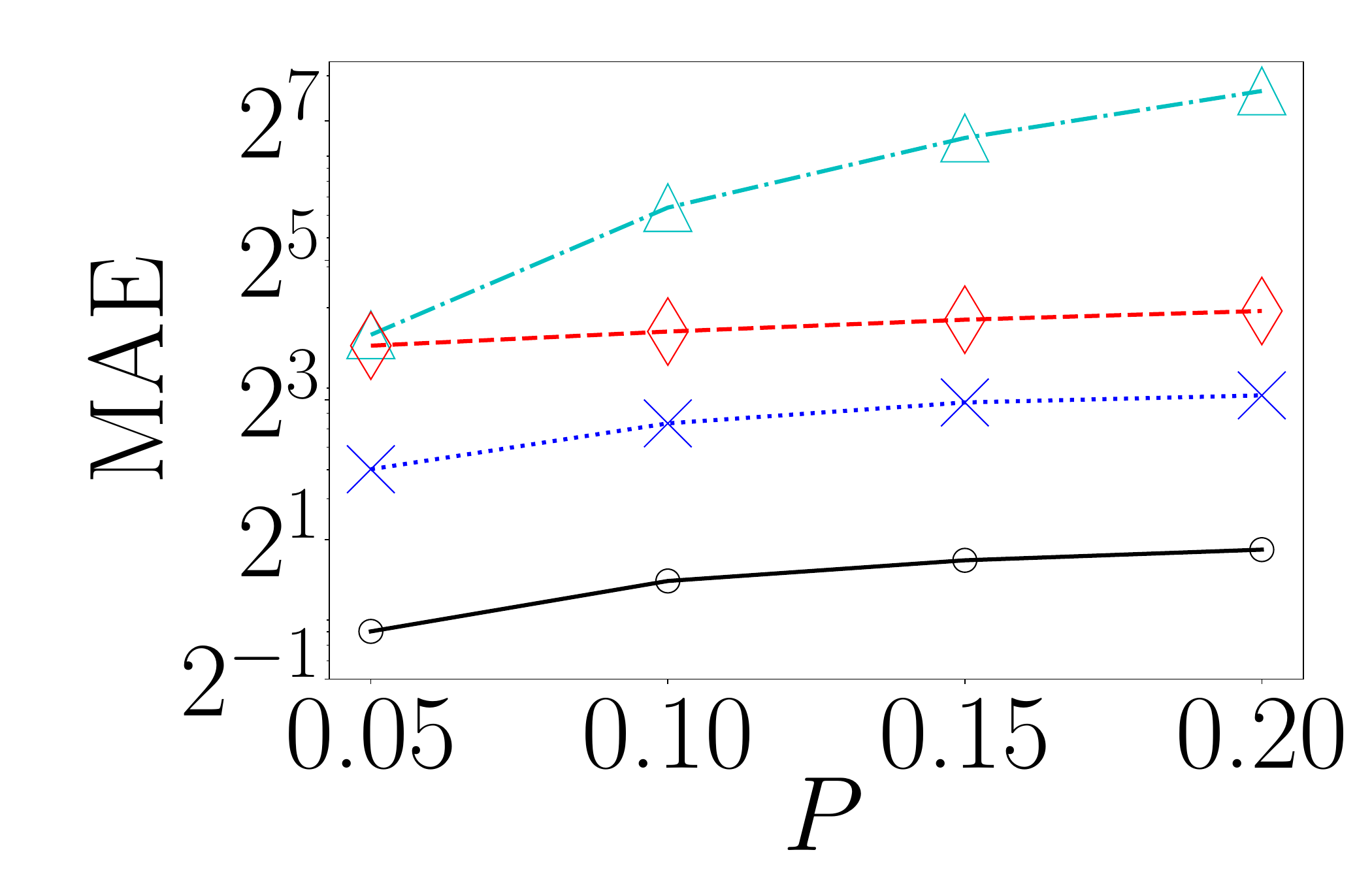}}
  \subfloat[MAE vs. $P$, Brightkite]{%
    \includegraphics[width=.25\linewidth]{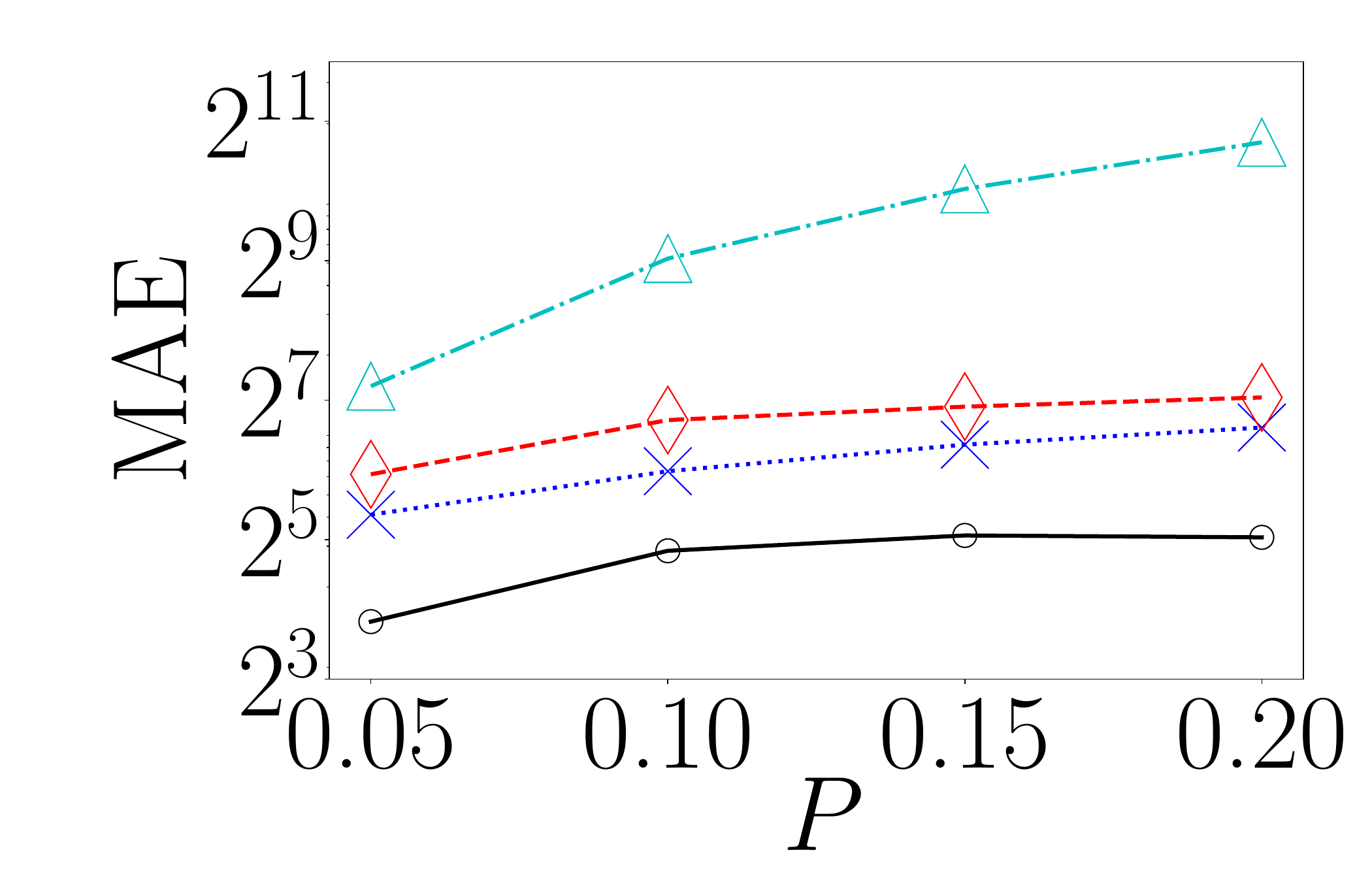}}
  \subfloat[MRE vs. $P$, Yelp]{%
    \includegraphics[width=.25\linewidth]{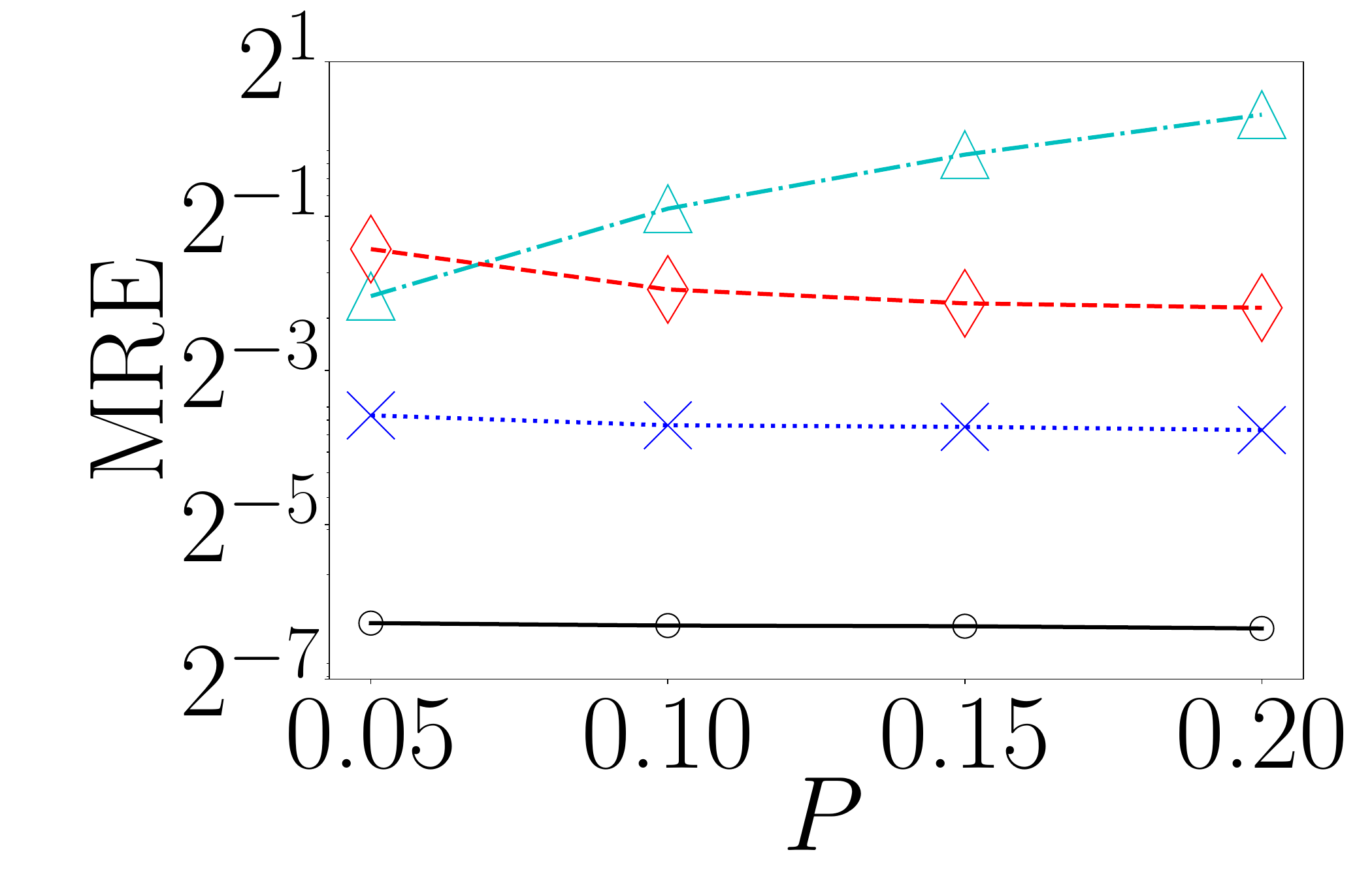}}
  \subfloat[MRE vs. $P$, Brightkite]{%
    \includegraphics[width=.25\linewidth]{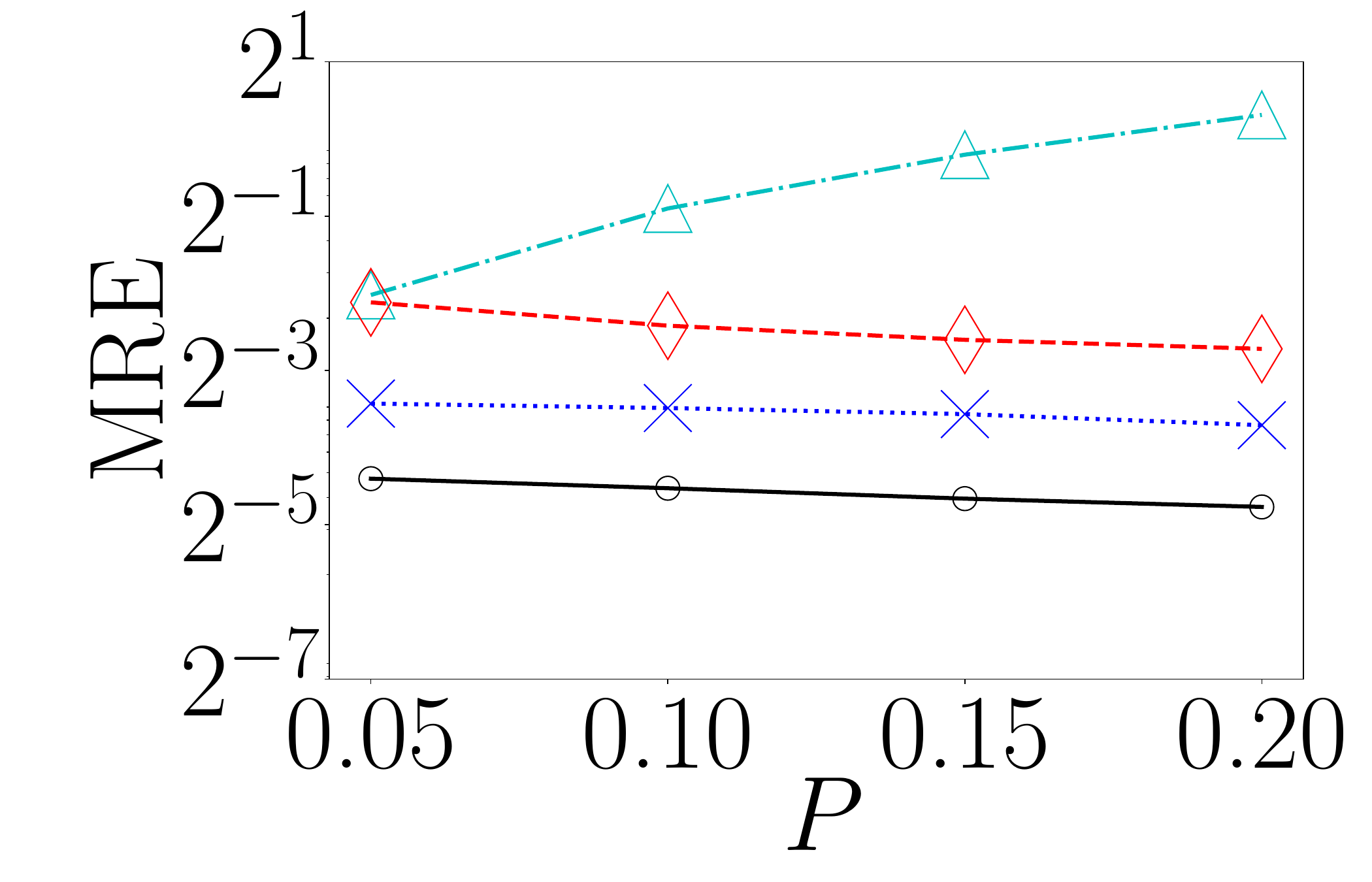}}\\
    \vspace{-0.15cm}
  \caption{Accuracy comparison with varying numbers of DHs ($I$) and proportions of the global dataset held by each DH ($P$).}
  \label{fig:result12}
\end{figure*}

\bullethdr{Accuracy Metrics.}
We use the Mean Absolute Error (MAE) and Mean Relative Error (MRE) to assess accuracy. Let $RC$ denote the ground-truth count and $\hat{RC}_j$ the estimated answer in the $j$-th run.
Then, the MAE and MRE can be calculated as: $\mathrm{MAE}=\frac{1}{J} \sum_{j=1}^J |\hat{RC}_j-RC|$ and $\mathrm{MRE}=\frac{1}{J} \sum_{j=1}^J \frac{|\hat{RC}_j-RC|}{RC}$.

\bullethdr{Efficiency Metrics.} We measure Time and Communication overhead. Time includes the end-to-end latency from generating the query range to obtaining the final plaintext range count.
Communication is the total data exchanged among the QU, CA, and DHs.

\vspace{0.1cm}
\noindent \textbf{Baselines}. We compare PPRC with two categories of representative baselines under comparable threat models.

\bullethdr{State-of-the-art protocols.}
These protocols are selected to represent the two major classes of prior privacy-preserving range counting solutions: encryption-based and MPC-based methods.
We include RCC~\cite{akhavan2023level}, which employs homomorphic encryption for bilateral privacy and securely aggregates partial counts from multiple DHs, and TVA~\cite{faisal2023tva}, an MPC-based protocol supporting exact range counting over overlapped datasets.
Since TVA produces exact results, its estimation error is omitted.

\bullethdr{Sketch-based variants.}
To evaluate different sketch estimators and aggregation schemes, we implement FM-Count and LL-Count by replacing the OLC component in PPRC with secure variants of FM~\cite{flajolet1985probabilistiC} and LogLog~\cite{durand2003loglog} sketches, respectively.
We also implement LC-Count, which adopts the same Linear Counting estimator~\cite{whang1990linear} as PPRC but uses secure bitwise-OR aggregation.
As LC-Count and PPRC employ identical estimators, they achieve the same accuracy, and we therefore omit the accuracy results of LC-Count.
The counter length of these sketches $S$ is selected from $\{1\mathrm{K},2\mathrm{K},3\mathrm{K},4\mathrm{K}\}$.

\vspace{0.1cm}
\noindent \textbf{Implementation.} 
We implement our PPRC using C++. 
The hash functions we used are MurmurHash~\cite{appleby2008murmurhash} with various random seeds.
By default, all experiments adopt the SHE scheme \cite{zhang2022efficient} as the homomorphic encryption primitive.
For SHE, we follow the parameter settings in~\cite{zhang2022efficient} and set $k_0=4096$, $k_1=80$, and $k_2=160$, where $k_0$ denotes the modulus size, $k_1$ is the plaintext length, and $k_2$ is the randomness parameter used during encryption.
The resulting 8192-bit public modulus ($2k_0$) provides at least 128-bit security against factorization attacks, while the 160-bit randomness ($k_2$) guarantees semantic security.
PPRC has bounded multiplicative depth and therefore does not require bootstrapping.
For baseline methods whose circuit depth exceeds this limit, ciphertext refreshing follows the bootstrapping protocol in~\cite{zheng2021efficient}.

Unless otherwise specified, Bloom filters use $K=7$ hash functions, satisfying the SHE noise budget.
The default false positive rate is set to $f_p=10^{-4}$ to balance accuracy and communication costs. The Bloom filter size $M$ is computed using the standard formula in \cref{subsec:bloom_filter}.
All experiments are conducted on a local area network (LAN) with 2 Gbit/s bandwidth and an average latency of 0.076 ms.
Each entity runs on a server with an Intel Xeon E5-2690 CPU (8 cores, 2.6GHz).
All results are averaged over 100 runs.

\begin{figure*}[ht]
  \centering
  \setlength{\abovecaptionskip}{0.3cm}
  \includegraphics[width=1.0\linewidth]{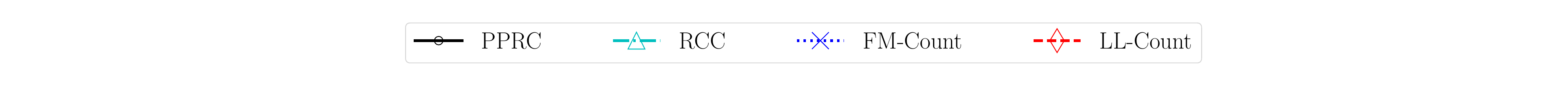}\\
  \vspace{-0.53cm}
    \subfloat[MAE vs. $S$, Yelp]{%
    \includegraphics[width=.25\linewidth]{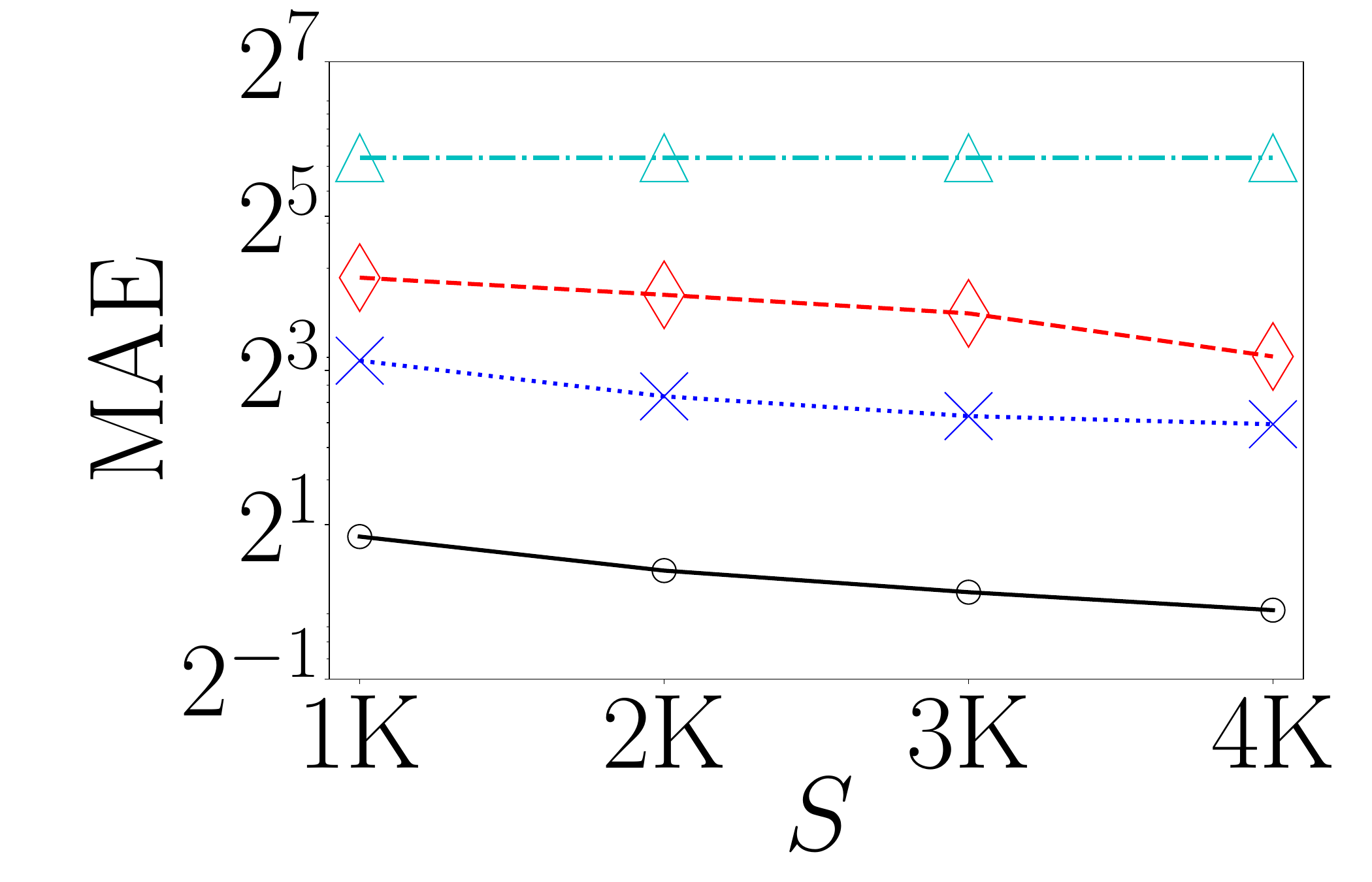}}
  \subfloat[MAE vs. $S$, Brightkite]{%
    \includegraphics[width=.25\linewidth]{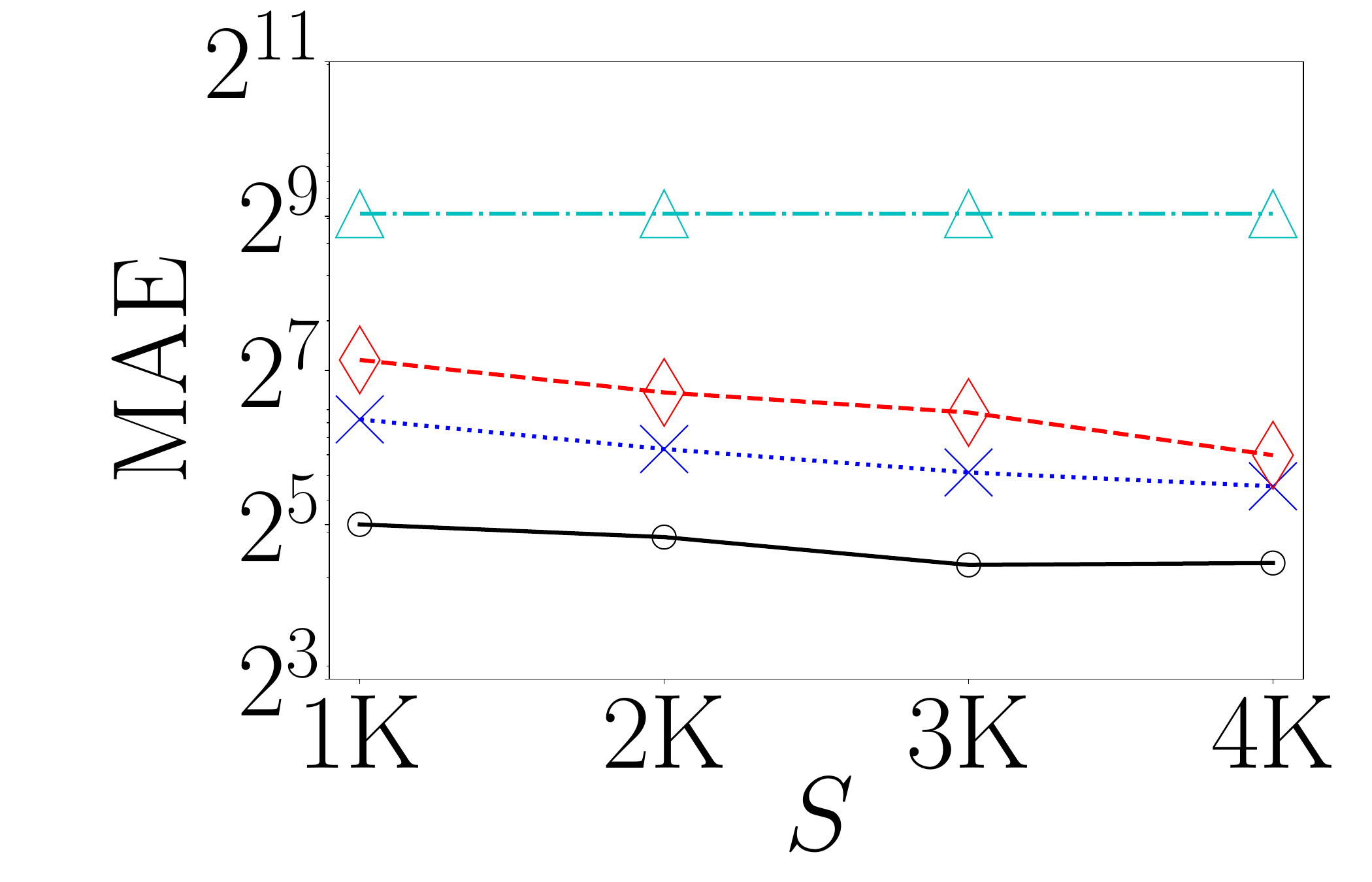}}
  \subfloat[MRE vs. $S$, Yelp]{%
    \includegraphics[width=.25\linewidth]{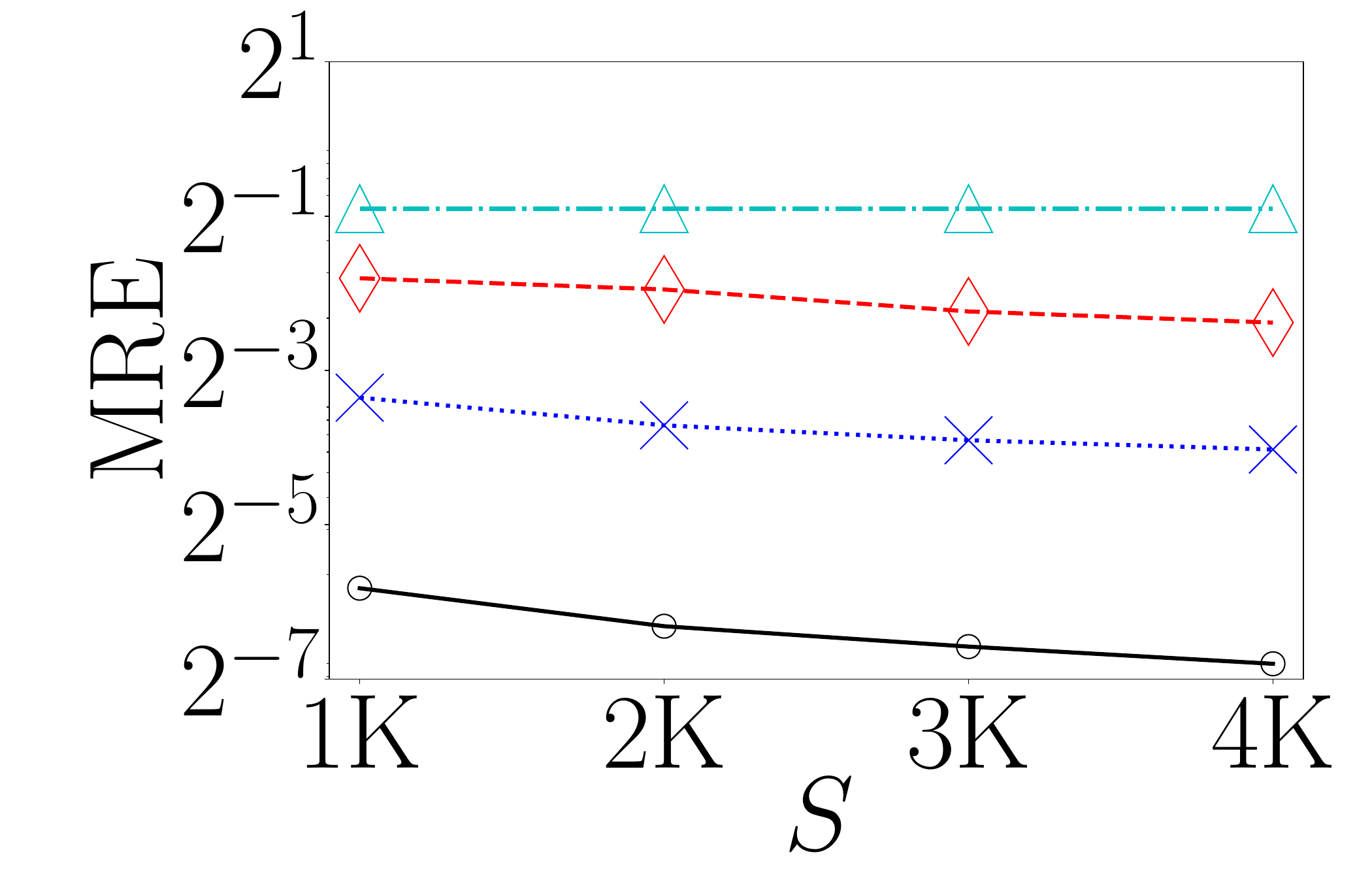}}
  \subfloat[MRE vs. $S$, Brightkite]{%
    \includegraphics[width=.25\linewidth]{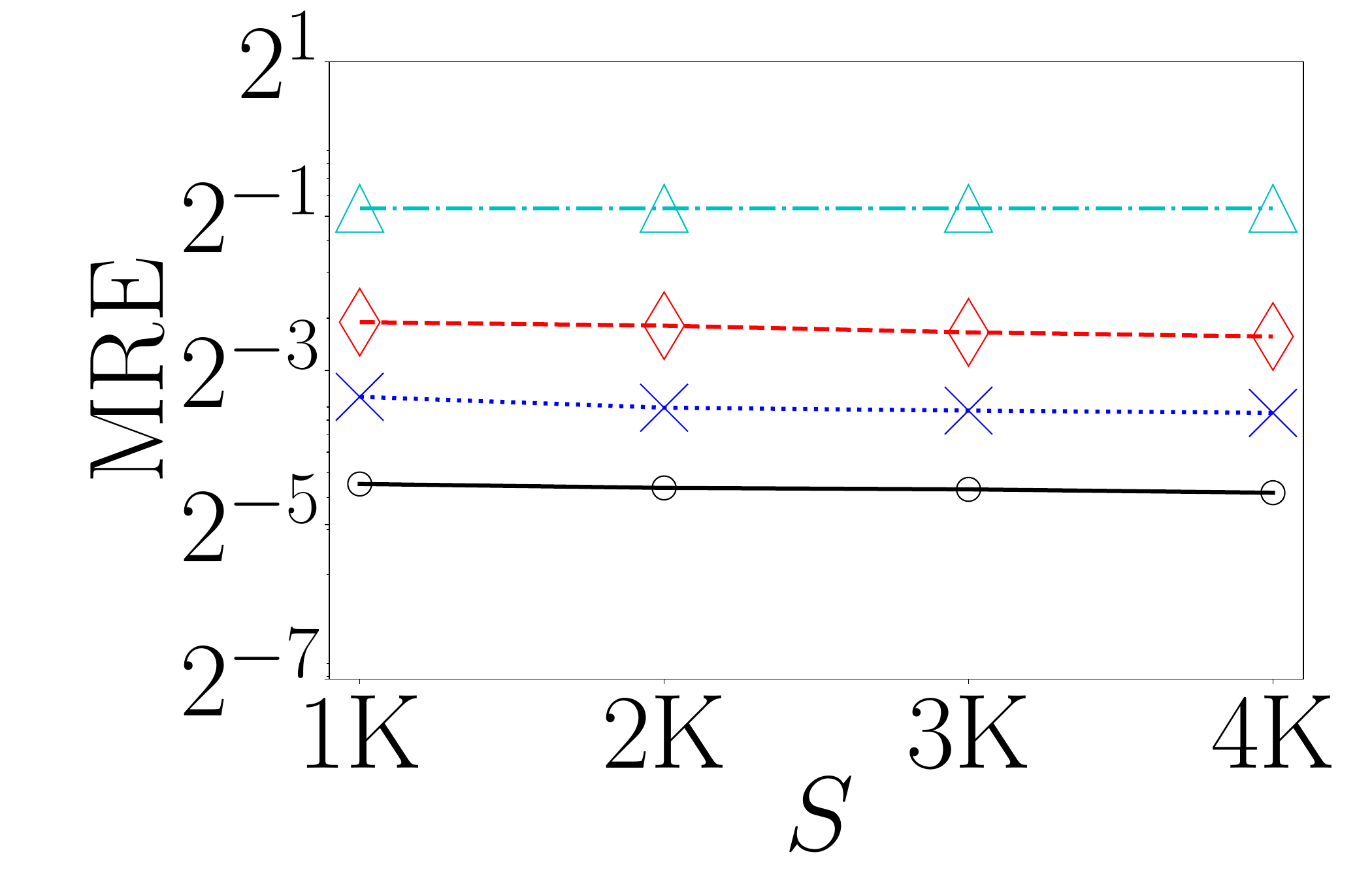}}
    \vspace{-0.15cm}
  \caption{Accuracy comparison with varying sketch counter lengths ($S$).}
  \label{fig:result13}
\end{figure*}

%% file: exp_results.tex
\subsection{Accuracy Comparison}

\begin{figure*}[ht]
  \centering
  \setlength{\abovecaptionskip}{0.3cm}
  \includegraphics[width=1.0\linewidth]{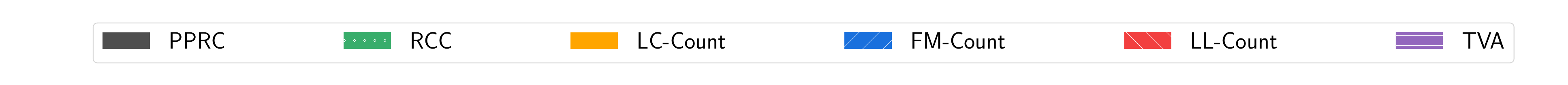}\\
  \vspace{-0.53cm}
  \subfloat[Time vs. $I$, $P=0.1$, Yelp]{%
    \includegraphics[width=.25\linewidth]{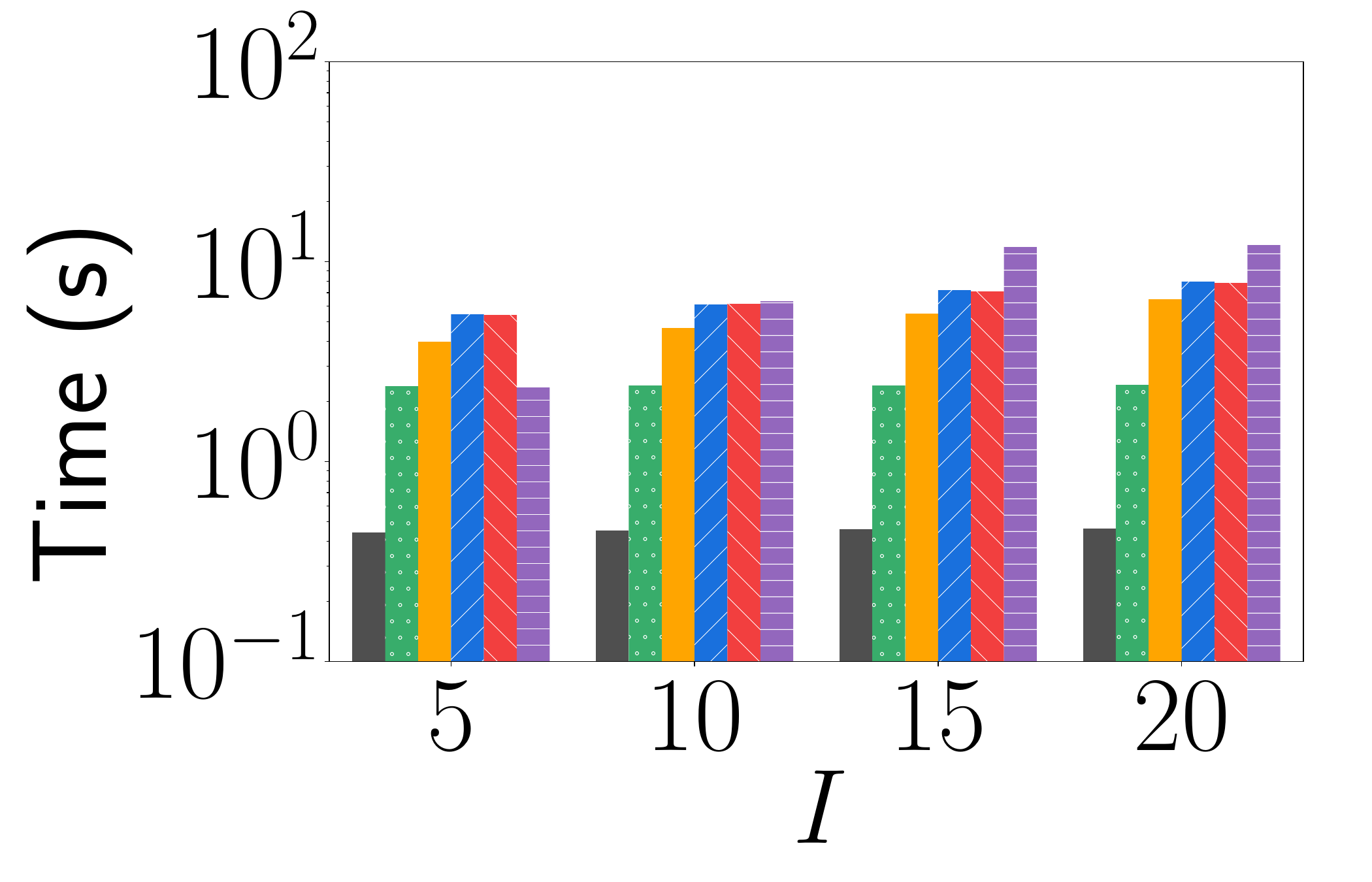}}
  \subfloat[Time vs. $I$, $P=0.1$, Brightkite]{%
    \includegraphics[width=.25\linewidth]{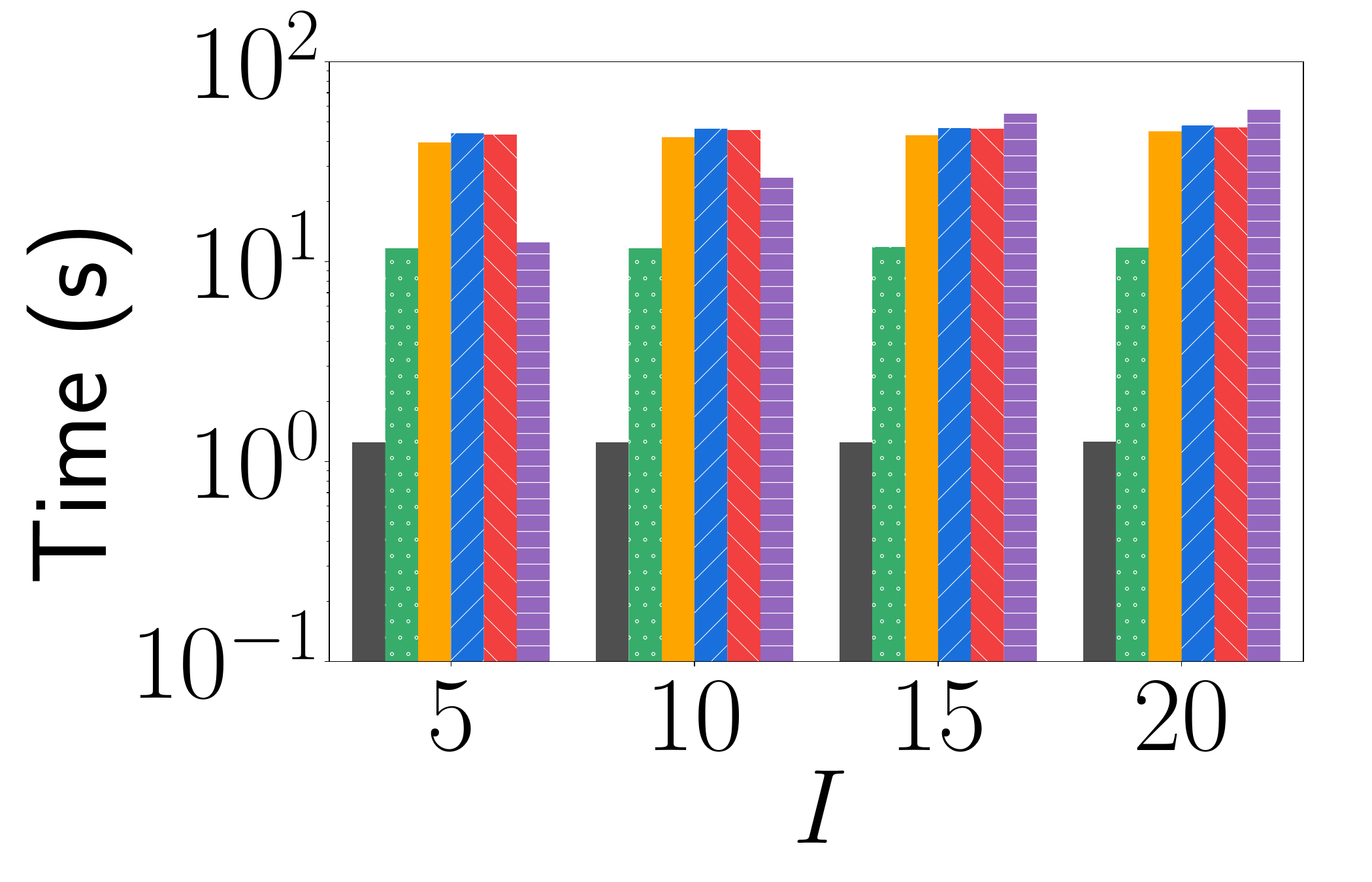}}
  \subfloat[Time vs. $P$, $I=10$, Yelp]{%
    \includegraphics[width=.25\linewidth]{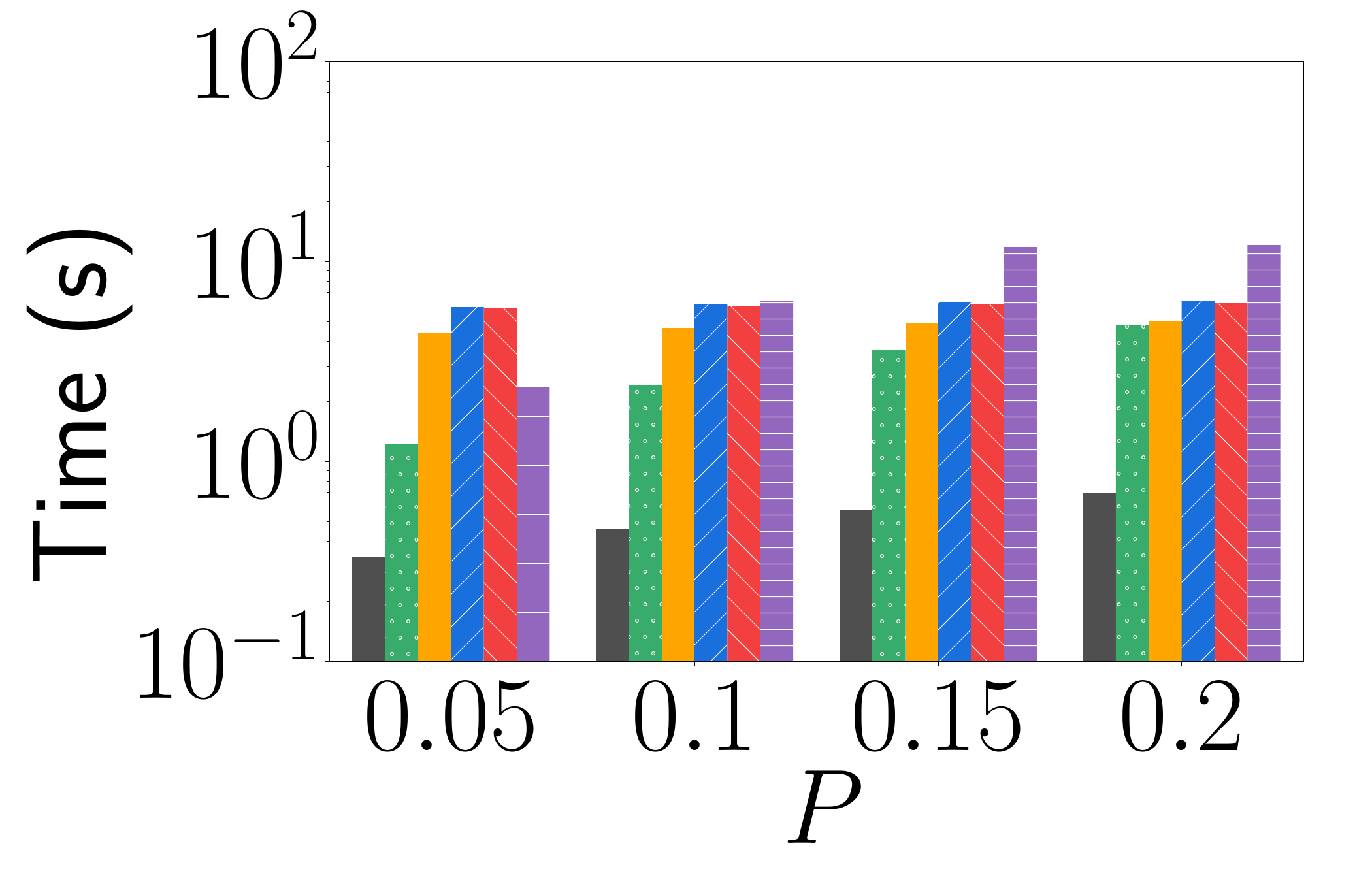}}
  \subfloat[Time vs. $P$, $I=10$, Brightkite]{%
    \includegraphics[width=.25\linewidth]{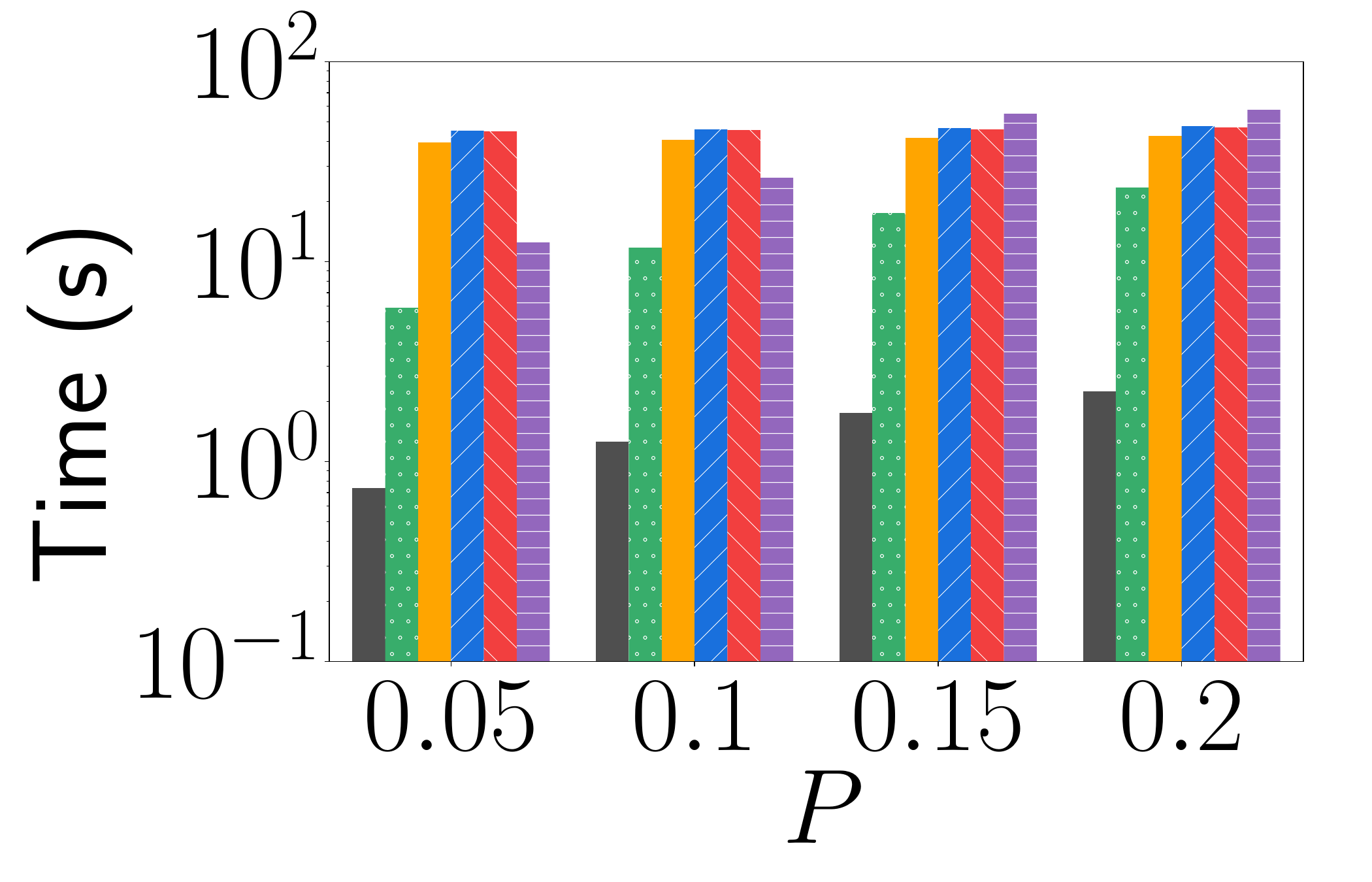}}
    \vspace{-0.15cm}
  \caption{Time cost comparison with various $I$ and $P$.}
  \label{fig:result21}
\end{figure*}

In this section, we compare the accuracy of PPRC with RCC, FM-Count, and LL-Count.
The results consistently show that PPRC is several times more accurate than the other protocols.

\noindent \textbf{\textit{The Effect of $I$.}}
\cref{fig:result12}(a)-(d) report MAE and MRE as the number of DHs varies over $I \in \{5,10,15,20\}$.
Here we set $P=0.1$ and $S=2 \mathrm{K}$.
As $I$ increases, MAE increases while MRE decreases for PPRC, FM-Count, and LL-count; for RCC, both the MAE and MRE increase.
Furthermore, PPRC is several times more accurate than RCC, FM-Count, and LL-Count.
For example, in \cref{fig:result12}(a) (Yelp dataset), PPRC reduces the MAE by an average of 55.79$\times$, 4.90$\times$, and 12.69$\times$ compared with RCC, FM-Count, and LL-Count, respectively. \cref{fig:result12}(b)-(d) corroborate these findings.

\noindent \textbf{\textit{The Effect of $P$.}}
\cref{fig:result12}(e)-(h) show MAE and MRE as $P$ varies over $\{0.05,0.1,0.15,0.2\}$, where $P$ denotes the proportion of each DH's private dataset relative to the global real-world dataset.
We fix $I=10$ and $S=2 \mathrm{K}$.
As $P$ increases, PPRC, FM-Count, and LL-Count exhibit increasing MAE but decreasing MRE, whereas RCC shows increases in both metrics.
Besides, our PPRC remains several times more accurate than RCC, FM-Count, and LL-Count.
\cref{fig:result12}(e) presents the MAE results for the Yelp dataset. 
On average, PPRC is 55.30, 4.79, and 12.65 times more accurate than RCC, FM-Count, and LL-Count.
The other results show consistent results.

\begin{figure*}[ht]
  \centering
  \setlength{\abovecaptionskip}{0.3cm}
  \includegraphics[width=1.0\linewidth]{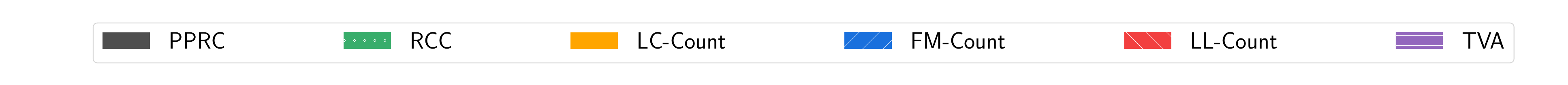}\\
  \vspace{-0.53cm}
  \subfloat[Comm vs. $I$, $P=0.1$, Yelp]{%
    \includegraphics[width=.25\linewidth]{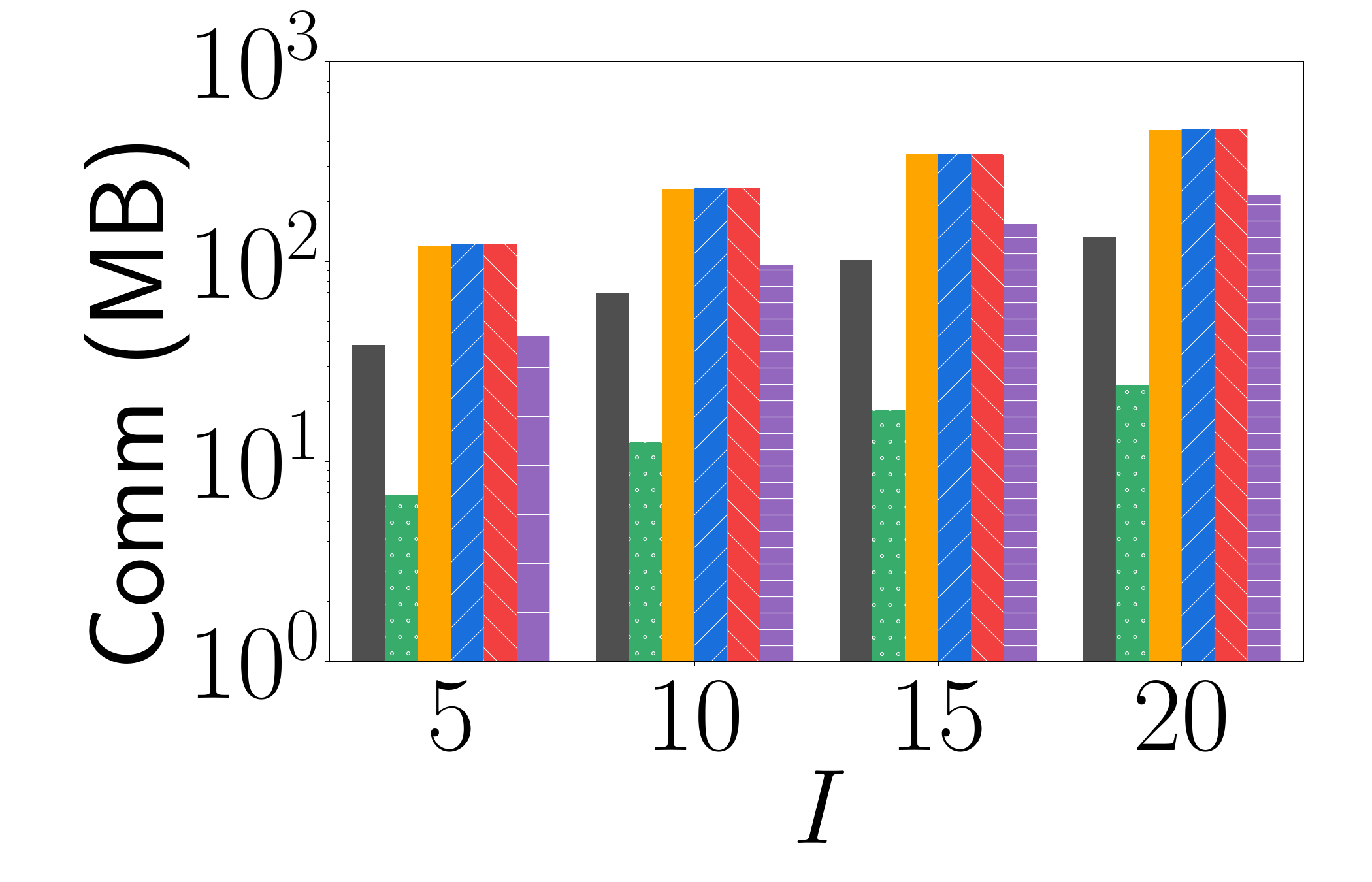}}
  \subfloat[Comm vs. $I$, $P=0.1$, Brightkite]{%
    \includegraphics[width=.25\linewidth]{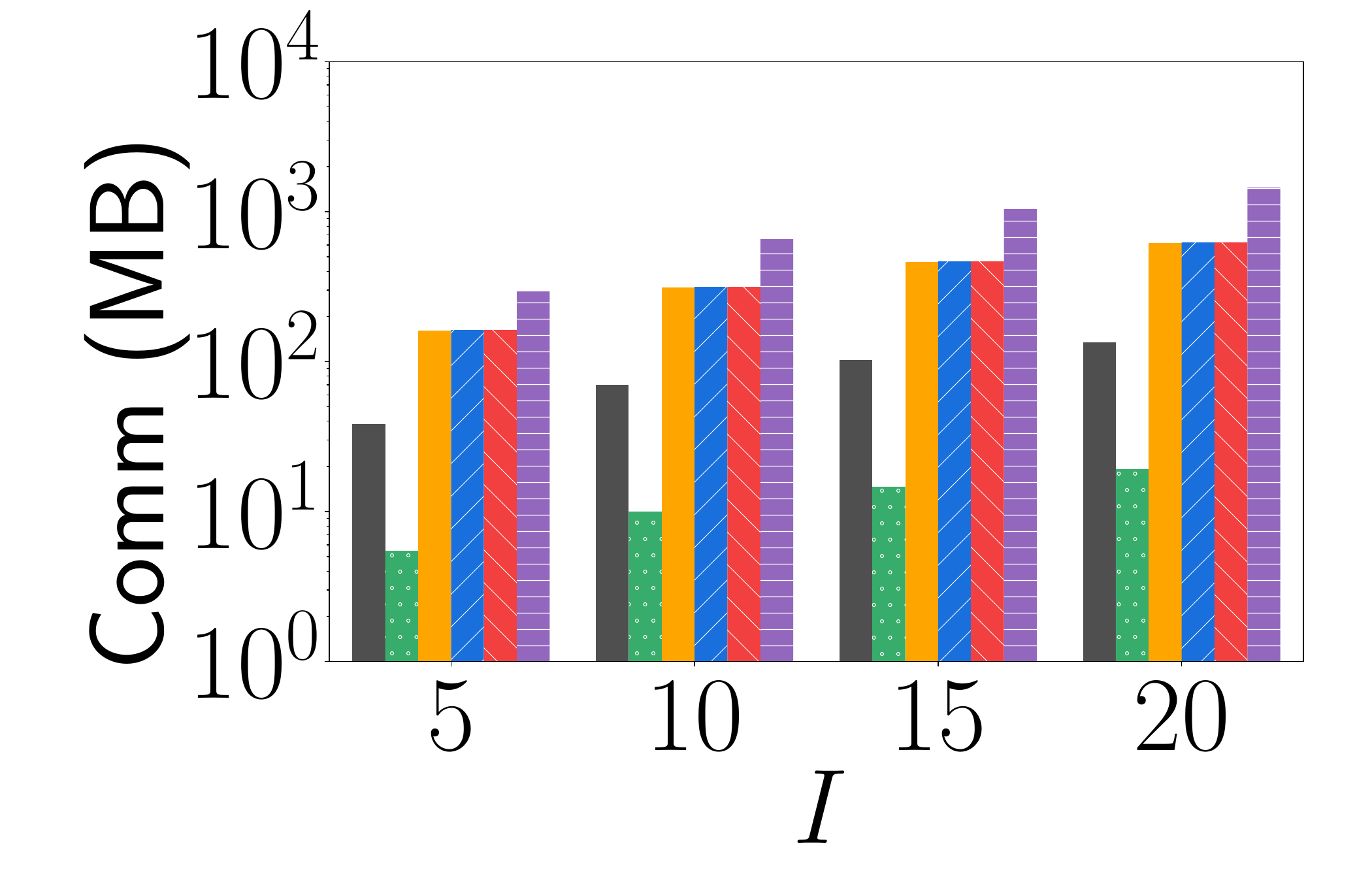}}
  \subfloat[Comm vs. $P$, $I=10$, Yelp]{%
    \includegraphics[width=.25\linewidth]{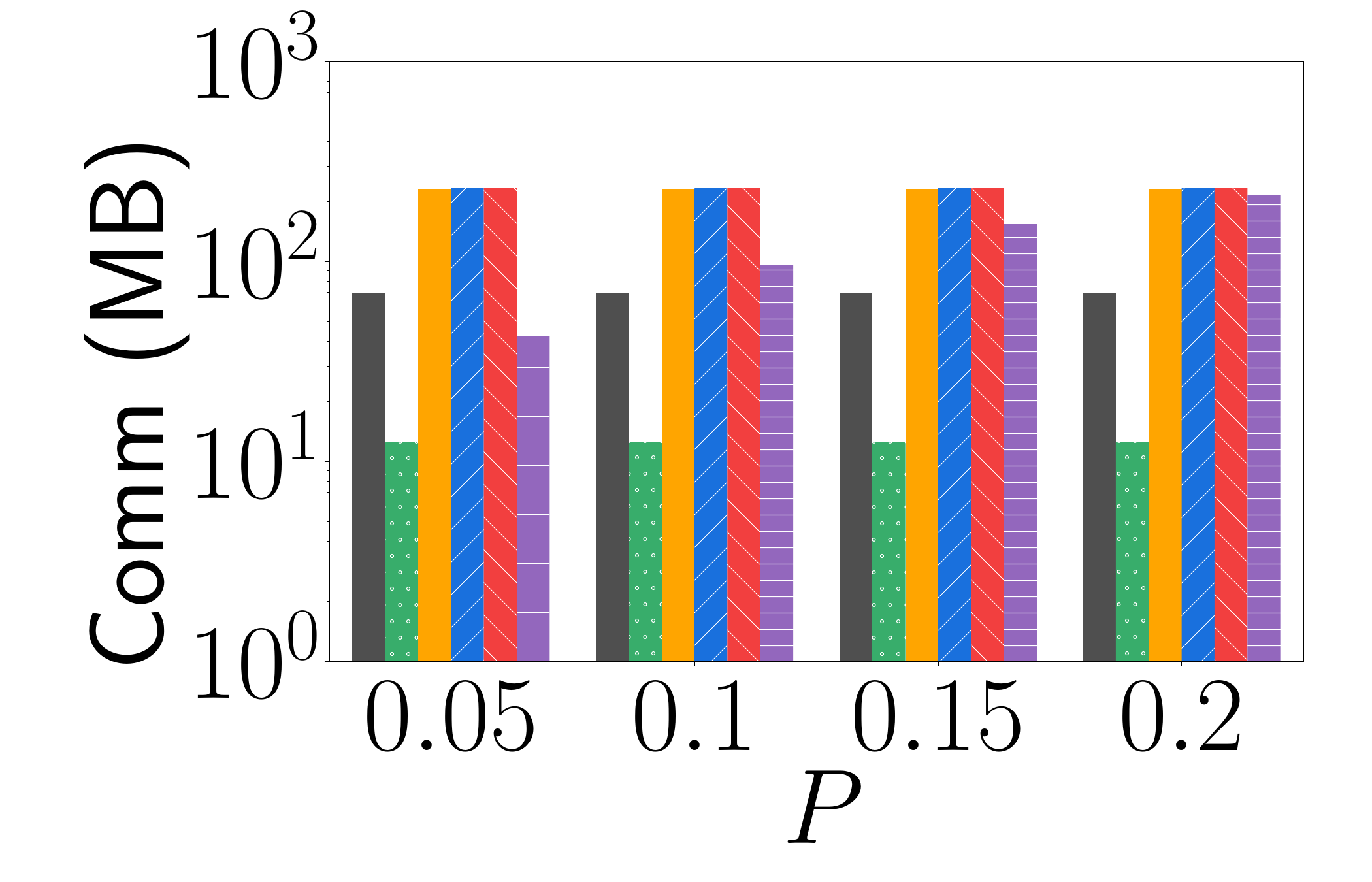}}
  \subfloat[Comm vs. $P$, $I=10$, Brightkite]{%
    \includegraphics[width=.25\linewidth]{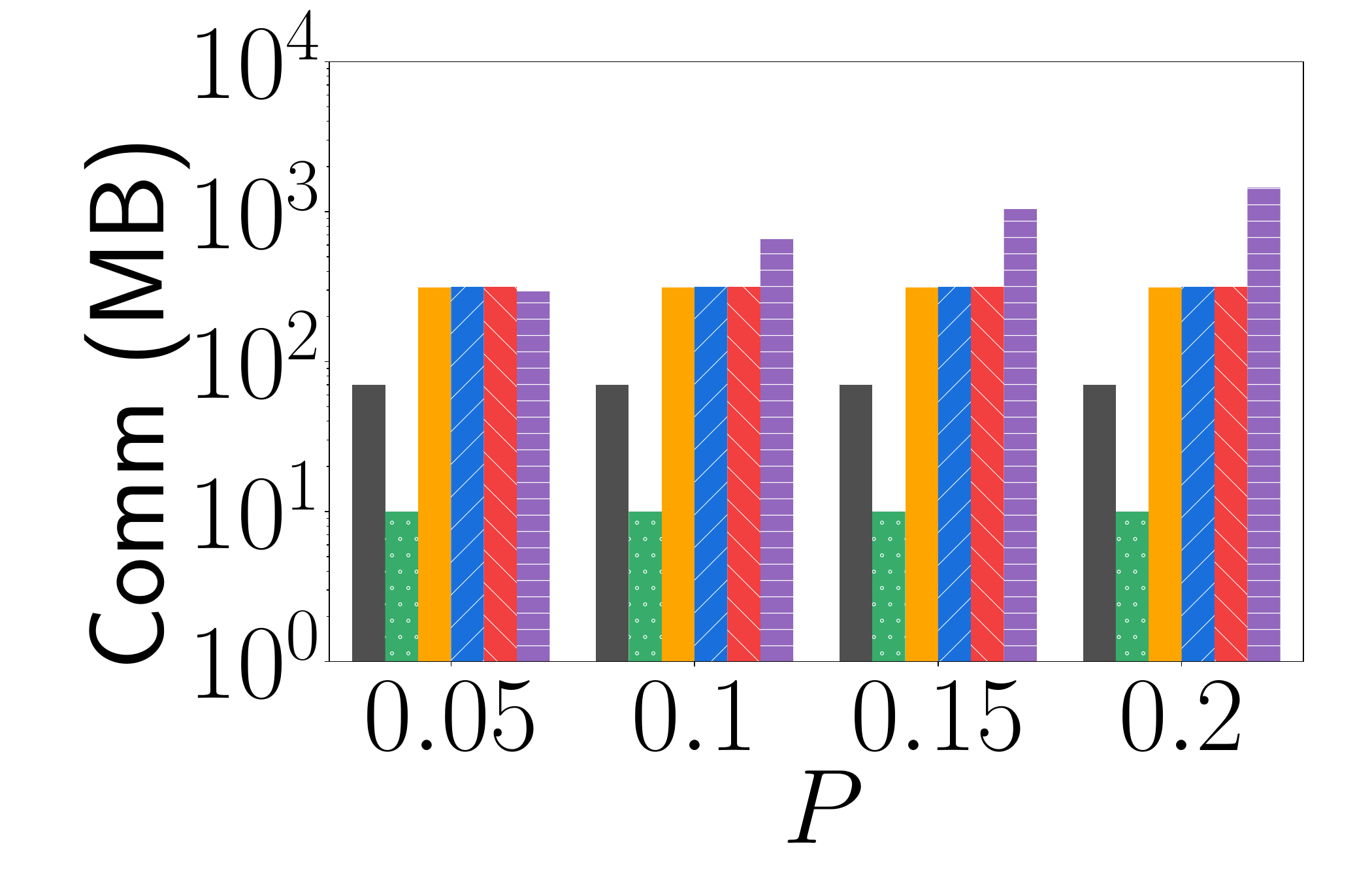}}
    \vspace{-0.15cm}
  \caption{Communication cost comparison with various $I$ and $P$.}
  \label{fig:result22}
\end{figure*}

\noindent \textbf{\textit{The Effect of $S$.}}
\cref{fig:result13}(a)-(d) evaluate the effect of the sketch counter length $S \in \{1 \mathrm{K}, 2 \mathrm{K}, 3 \mathrm{K}, 4 \mathrm{K}\}$ under $P=0.1$ and $I=10$.
As expected, increasing $S$ reduces errors for all sketch-based protocols.
The results for RCC remain constant, as RCC does not use sketches.
PPRC is several times more accurate than RCC, FM-Count, and LL-Count.
As shown in \cref{fig:result13}(a), our PPRC is on average 44.61, 4.95, and 11.03 times more accurate than RCC, FM-Count, and LL-Count.
Similar trends are observed in \cref{fig:result13}(b)–(d).

\subsection{Efficiency Comparison}
This section evaluates the efficiency of PPRC, RCC, LC-Count, FM-Count, LL-Count, and TVA. 
The results show that PPRC is dozens of times more efficient than the other protocols.

\noindent \textbf{\textit{Time Cost.}}
\cref{fig:result21}(a)-(d) report the Time cost with respect to the number of DHs $I \in \{5, 10, 15, 20\}$ and the data proportion $P \in \{0.05,0.1,0.15,0.2\}$.
We set $I=10$, $P=0.1$, and $S=2 \mathrm{K}$ as default.
As $I$ and $P$ increase, the Time cost of all protocols grows.
Besides, PPRC is dozens of times faster than RCC, LC-Count, FM-Count, LL-Count, and TVA.
For example, in \cref{fig:result21}(b) (Brightkite dataset), PPRC is on average 9.37, 33.89, 36.87, 36.36, and 30.16 times faster than RCC, LC-Count, FM-Count, LL-Count, and TVA, respectively.
Consistent trends are shown in \cref{fig:result21}(a), (c), and (d).

\noindent \textbf{\textit{Communication Cost.}}
\cref{fig:result22}(a)-(d) presents the Communication cost concerning $I \in \{5,10,15,20\}$ and $P \in \{0.05,0.1,0.15,0.2\}$.
We also set $I=10$, $P=0.1$, and $S=2 \mathrm{K}$ as default.
As $I$ increases, the communication cost of all protocols increases. As $P$ increases, only TVA incurs higher communication overhead, while the costs of the other protocols remain unchanged due to their use of fixed-size encrypted sketches or fixed-format query messages.
Furthermore, the communication cost of PPRC is several times lower than that of LC-Count, FM-Count, LL-Count, and TVA.
Specifically, \cref{fig:result22}(b) presents the Communication cost for the Brightkite dataset with various $I$.
On average, PPRC reduces the communication cost by 4.75, 4.79, 4.75, and 10.50 times compared with LC-Count, FM-Count, LL-Count, and TVA.
Although RCC incurs lower communication costs, it is less accurate and computationally less efficient than PPRC.
Consistent results are shown in \cref{fig:result22}(a), (c), and (d).

\noindent \textbf{\textit{The Effect of $S$.}}
\cref{fig:result23}(a) evaluates the Time costs across varying sketch counter lengths $S \in \{1 \mathrm{K}, 2 \mathrm{K}, 3 \mathrm{K}, 4 \mathrm{K}\}$, with $I=10$ and $P=0.1$.
Since RCC and TVA do not use count estimation sketches, their costs remain constant regardless of $S$.
While the Time cost of all sketch-based protocols increases with $S$.
Furthermore, PPRC is dozens of times faster than the other protocols.
In \cref{fig:result23}(a), PPRC is on average 9.29, 44.92, 44.79, 44.57, and 20.86 times faster than RCC, LC-Count, FM-Count, LL-Count, and TVA.

\begin{figure}[t]
  \centering
  \setlength{\abovecaptionskip}{0.3cm}
  \includegraphics[width=1.0\linewidth]{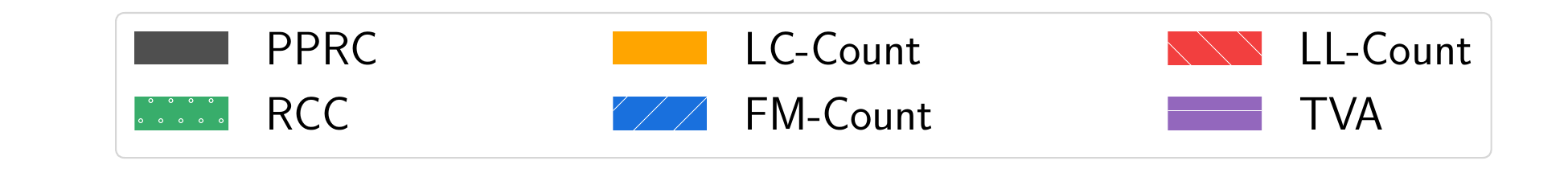}\\
  \vspace{-0.41cm}
  \subfloat[Time vs. $S$, Brightkite]{%
    \includegraphics[width=.5\linewidth]{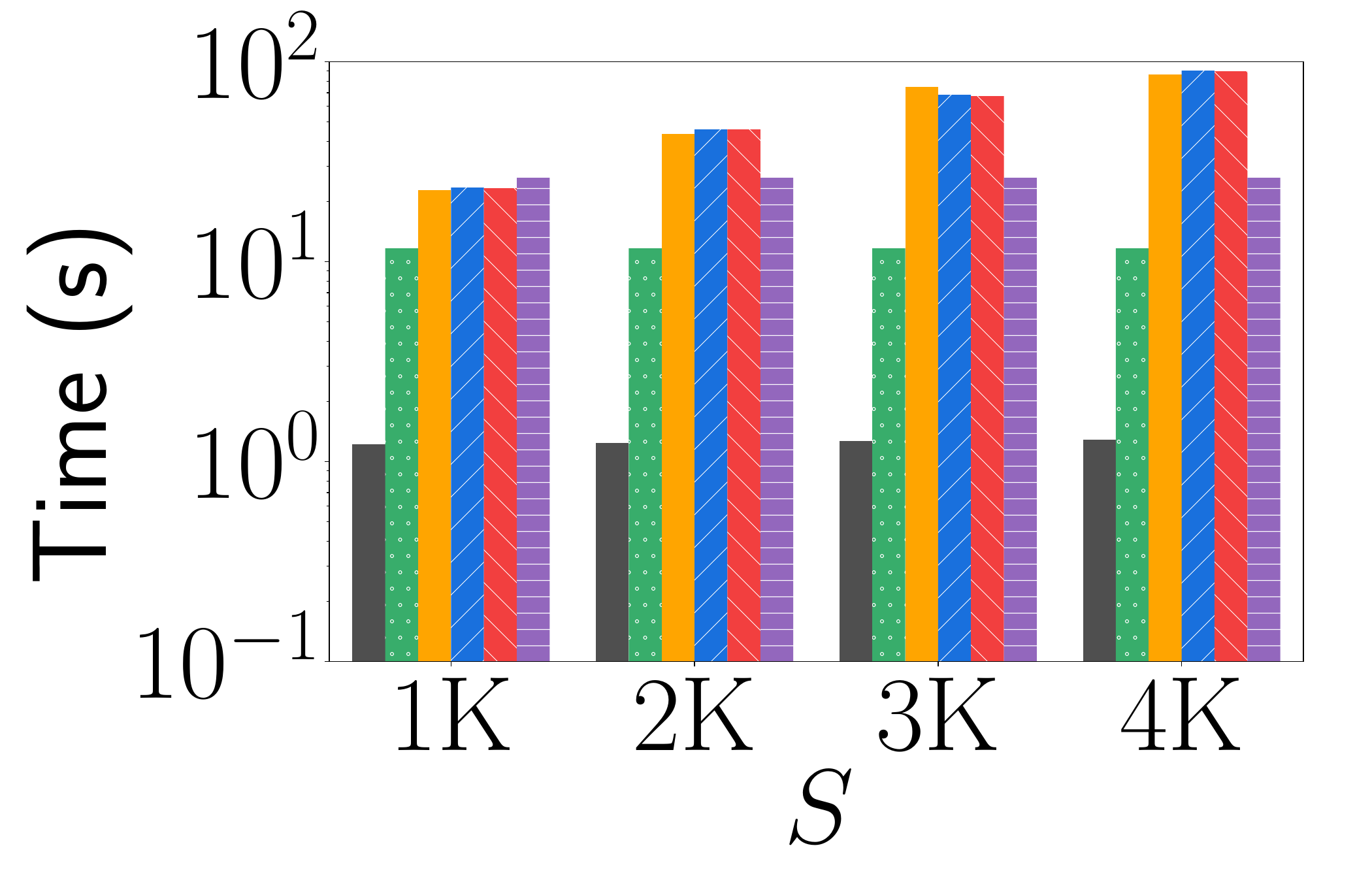}}
  \subfloat[Time vs. $|U|$, Synthetic dataset]{%
    \includegraphics[width=.5\linewidth]{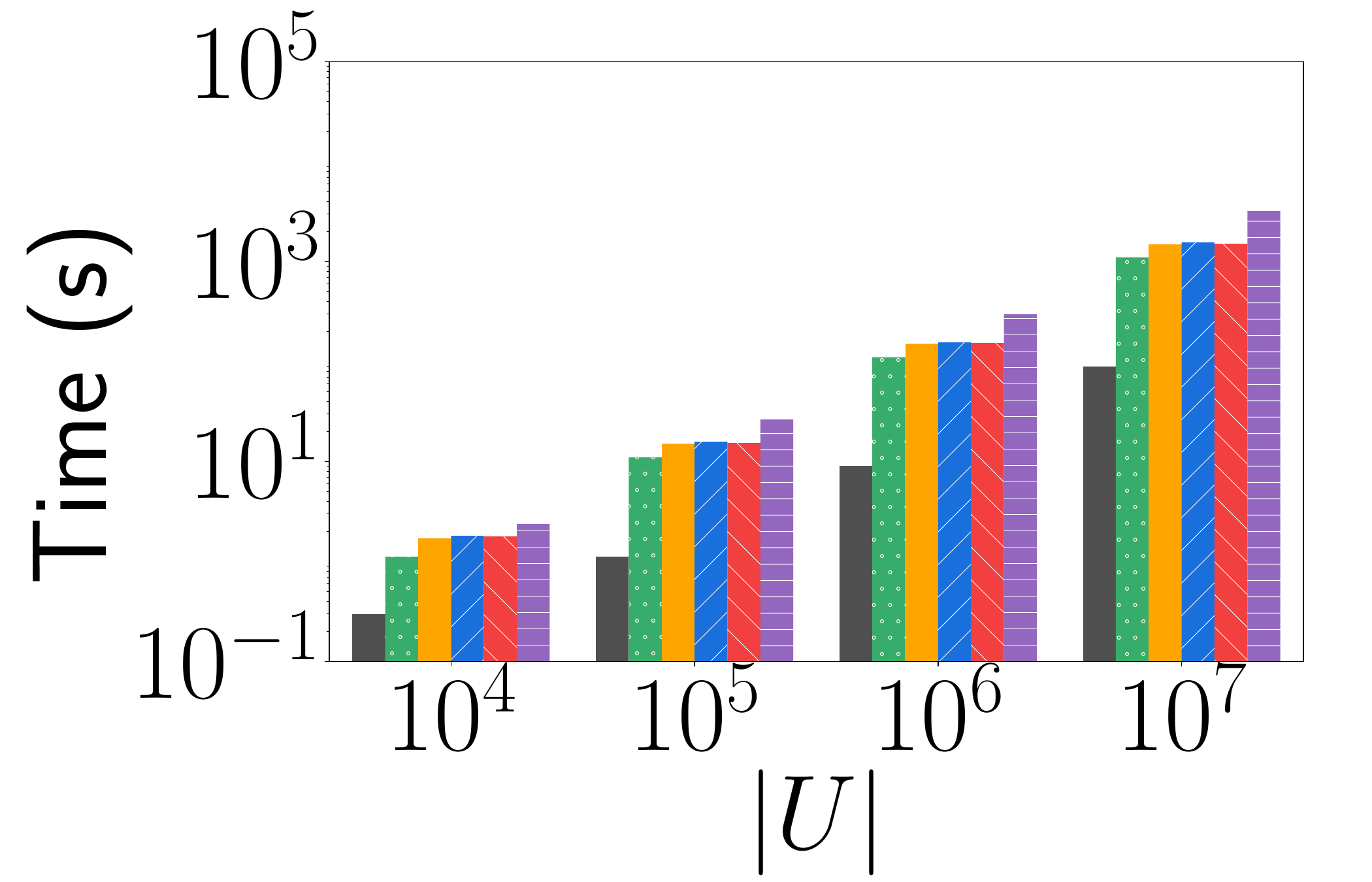}}
    \vspace{-0.15cm}
  \caption{Time cost comparison with various $S$ and $|U|$.}
  \label{fig:result23}
\end{figure}

\noindent \textbf{\textit{Scalability Evaluation.}}
We evaluate the scalability of our PPRC on synthetic datasets with varying total numbers of data records $|U| \in \{10^4, 10^5, 10^6, 10^7\} $. Here, we fix $I=10$, $P=0.1$, and $S=2 \mathrm{K}$. 
Since the communication cost consists of fixed-size sketches and remains constant regardless of $|U|$, we report only the Time cost.

\cref{fig:result23}(b) shows that the Time cost of PPRC increases linearly with $|U|$, aligning with theoretical analysis.
In addition, the Time cost of PPRC is dozens of times lower than that of other protocols.
For $ |U| = 10^7$, the Time cost of PPRC is only 88.733\,s.
On average, PPRC is 9.47, 13.21, 13.65, 13.33, and 25.07 times faster than RCC, LC-Count, FM-Count, LL-Count, and TVA, respectively.

\noindent \textbf{\textit{Performance under BFV Encryption.}}
To demonstrate protocol generality, we replace the underlying SHE scheme with the widely used BFV scheme~\cite{2012somewhat} provided by the Microsoft SEAL library~\cite{chen2017}.
This scheme uses an 8192-degree polynomial modulus and a 218-bit ciphertext modulus to ensure 128-bit security.
For fair comparison, RCC adopts the same BFV configuration.
The number of hash functions is set to $K=3$ to satisfy the noise budget constraint. 
We conduct experiments on the Brightkite dataset with $I=20$, $S=1 \mathrm{K}$, and $P \in \{0.05, 0.1, 0.15, 0.2\}$, as shown in \cref{fig:seal_bfv_performance}.

As shown in \cref{fig:seal_bfv_performance}(a), PPRC achieves the lowest runtime across all $P$ (7--9 seconds).
This efficiency is mainly due to the batching capability of BFV, enabling SIMD-style evaluation over thousands of packed records within a single ciphertext.
Besides, PPRC is up to $12.4\times$ faster than RCC and $17.7\times$ faster than TVA.
\cref{fig:seal_bfv_performance}(b) reports the communication overhead.
TVA incurs the highest communication cost.
Although RCC requires lower communication than PPRC, PPRC maintains bounded communication cost below 500\,MB while achieving considerably higher accuracy and better computational efficiency.
These results confirm that the efficiency advantage of PPRC does not rely on a specific FHE instantiation.

\begin{figure}[t]
  \centering
  \setlength{\abovecaptionskip}{0.3cm}
  \includegraphics[width=1.0\linewidth]{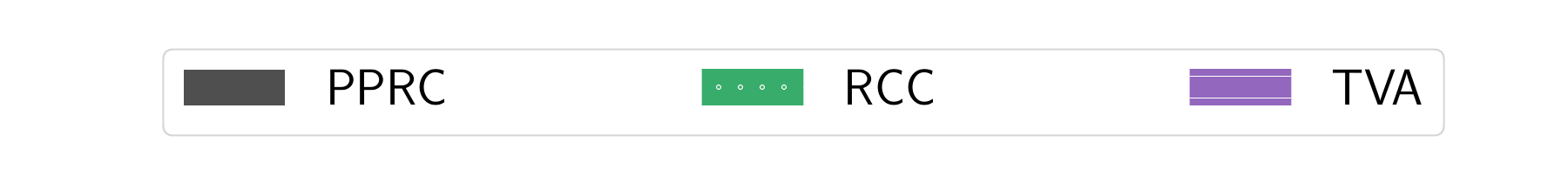}\\
  \vspace{-0.41cm}
  \subfloat[Time vs. $P$]{%
    \includegraphics[width=.5\linewidth]{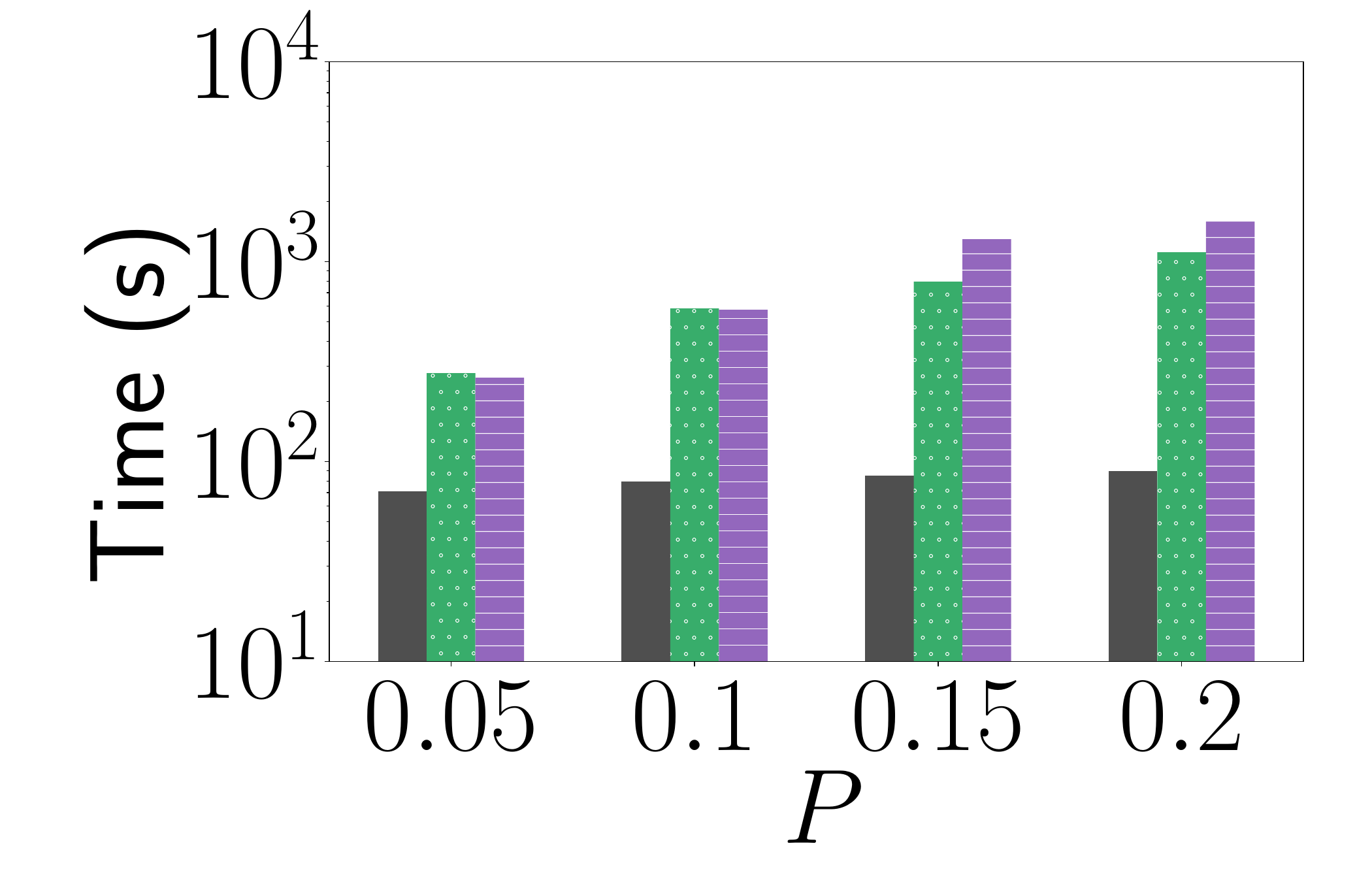}}
  \subfloat[Comm vs. $P$]{%
    \includegraphics[width=.5\linewidth]{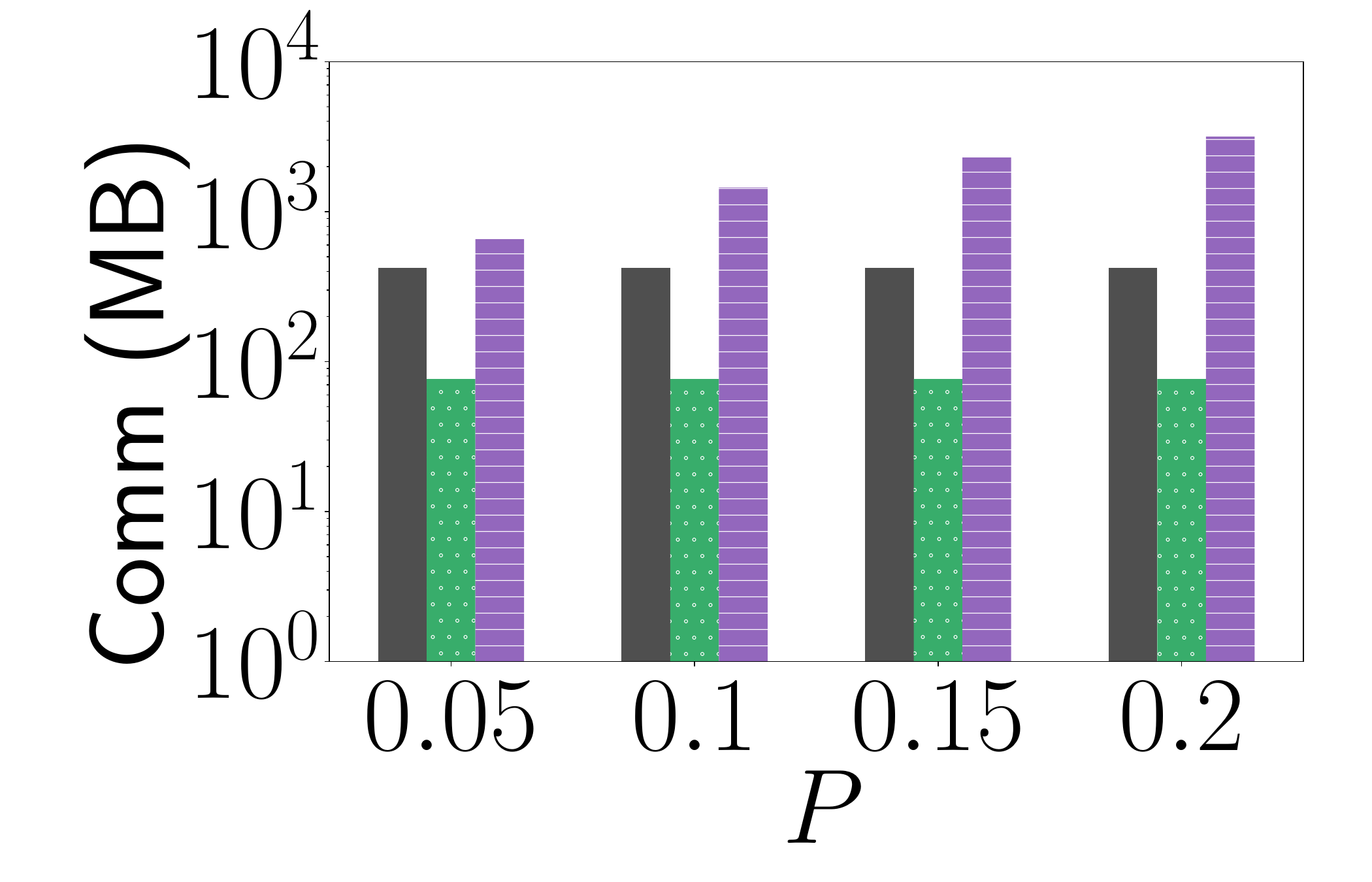}}
    \vspace{-0.15cm}
  \caption{Time and communication costs under various sampling proportions ($P$), using the BFV implementation in the Microsoft SEAL library.}
  \label{fig:seal_bfv_performance}
\end{figure}

\subsection{Parameter Analysis} \label{subsec:prange}
We utilize the Brightkite dataset to evaluate how the query range size and the Bloom filter's false positive rate ($f_p$) impact the performance of PPRC.

\noindent \textbf{\textit{Impact of the Query Range Size.}}
In PRP, larger query ranges require larger Bloom-filter encodings, which increases computation and communication overhead. We therefore evaluate how query range size affects overall system performance.
\cref{fig:result3}(a) reports the Time and communication costs as the query range side length expands from $1\,\mathrm{km}$ to $5\,\mathrm{km}$ (i.e., area from $1\,\text{km} \times 1\,\text{km}$ square to $5\,\text{km} \times 5\,\text{km}$).
PPRC demonstrates remarkable performance stability with respect to the query range size: the Time cost increases only slightly from $1.14$\,s to $1.50$\,s, even as the range expands by $25\times$.
This is because the computational complexity of PPRC is primarily dominated by operations dependent on the number of data records, regardless of the query range size.
Meanwhile, the communication cost grows from $46.1$\,MB to $142.3$\,MB. 
This is expected, since the encrypted Bloom filter size increases from 1.09\,MB to 5.47\,MB as the query range side length expands from 1\,km to 5\,km, resulting in higher communication costs.

\noindent \textbf{\textit{Impact of False Positive Rate ($f_p$).}}
\cref{fig:result3}(b) illustrates the trade-off between accuracy (measured by MAE) and communication cost when varying $f_p$ from $10^{-2}$ to $10^{-6}$.
A lower $f_p$ implies higher accuracy but requires larger Bloom filters, leading to higher communication cost.
We observe a sharp decline in MAE as $f_p$ decreases from $10^{-2}$ to $10^{-4}$, where the error drops drastically from $192.6$ to $28.7$.
However, further reducing $f_p$ beyond $10^{-4}$ yields diminishing returns. 
For instance, decreasing $f_p$ from $10^{-5}$ to $10^{-6}$ only reduces the MAE marginally (from $16.6$ to $15.3$), yet the communication cost continues to rise steadily.
Therefore, we select $f_p = 10^{-4}$ as the optimal default parameter, as it strikes the best balance between high accuracy and low communication cost.

%% file: related.tex
\section{Related Work}\label{sec:related}

Private geographic query (PGQ) schemes can be classified into two categories: outsourcing-data-based (\textbf{OD-PGQ}) and distributed-data-based (\textbf{DD-PGQ}).
In OD-PGQ, DHs encrypt and outsource their datasets to a third party, which processes user queries over the protected data.
In DD-PGQ, multiple DHs retain their datasets locally and jointly process user queries.

\textit{OD-PGQ.} Early works like CryptDB \cite{popa2011cryptdb} utilize Order-Preserving Encryption (OPE) to encrypt the plaintext, which is later shown to be vulnerable to inference attacks \cite{wang2013secure}.
Subsequent works improve index structures for secure range queries~\cite{wang2013secure, karras2016adaptive, li2015fast}.
With the development of location-based services, \cite{cui2019geo,wang2020search} design secure schemes to support geographic keyword queries. However, they are limited to the Boolean keyword test, i.e., determining whether a query keyword set is fully contained in a given keyword set. To address this, Song et al. \cite{song2021privacy} and Zhang et al. \cite{zhang2022efficient} design secure schemes supporting geographic keyword similarity queries based on Euclidean similarity and Jaccard similarity, respectively. 
TVA~\cite{faisal2023tva} supports deduplication to the encrypted database and thus supports more complicated queries. 
More recently, Zhang et al.~\cite{ZhangLZZGWSL24} improve the flexibility by supporting arbitrary query ranges.
Beyond encrypted databases, OD-PGQ systems also utilize Function Secret Sharing (FSS) \cite{dauterman2022waldo} and trusted execution environments like Intel SGX \cite{bittau2017prochlo,zhang2024fedknn}.
However, these systems typically rely on outsourced computation with non-colluding servers or trusted hardware.
In contrast, our setting assumes that each DH retains its data locally, which is common in cross-organizational applications where raw datasets cannot be outsourced due to privacy or regulatory constraints.

\begin{figure}[t]
  \centering
  \subfloat[Results for different range sizes]{%
    \includegraphics[width=.5\linewidth]{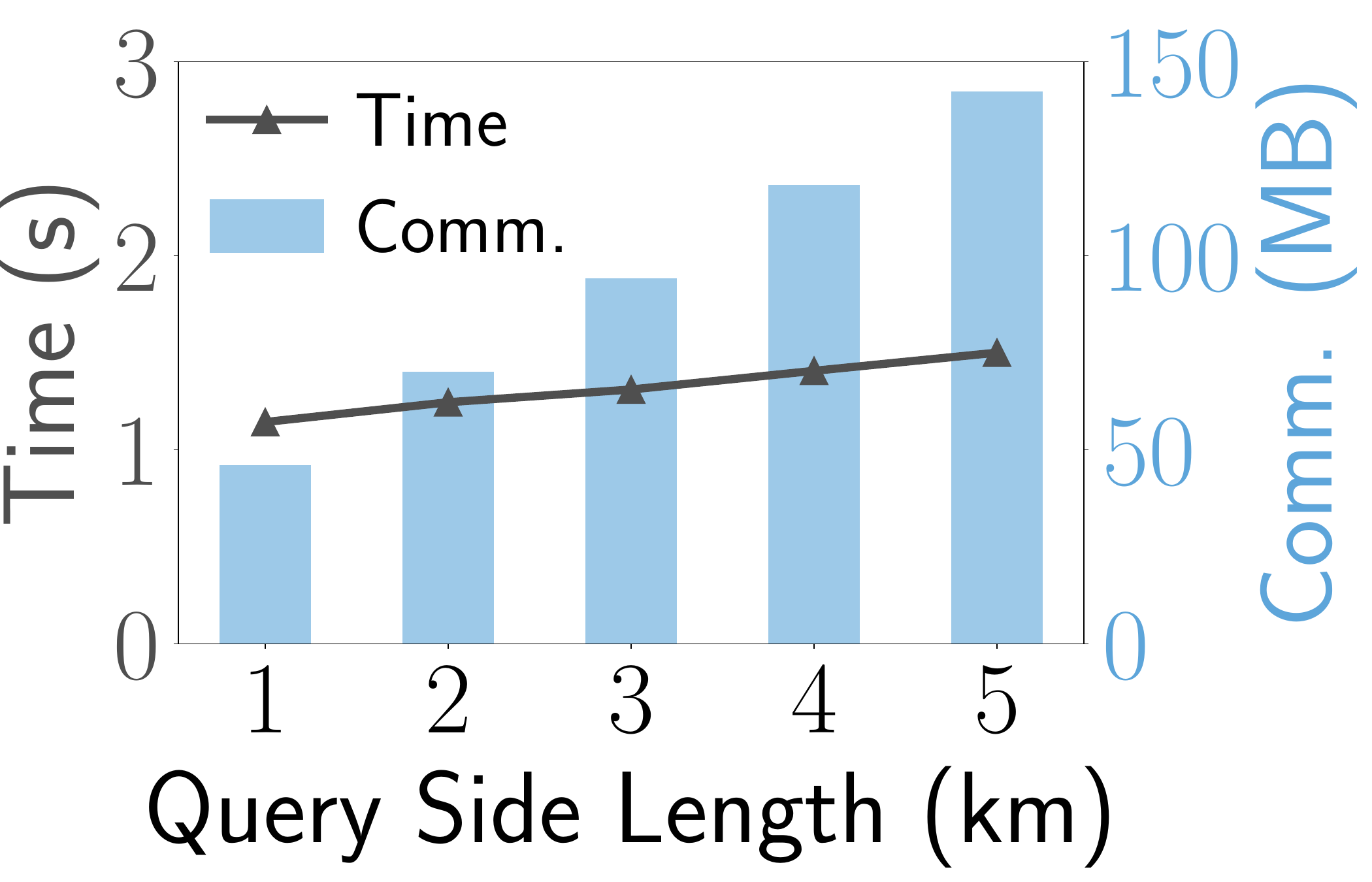}}
  \subfloat[Results for different $f_p$]{%
    \includegraphics[width=.5\linewidth]{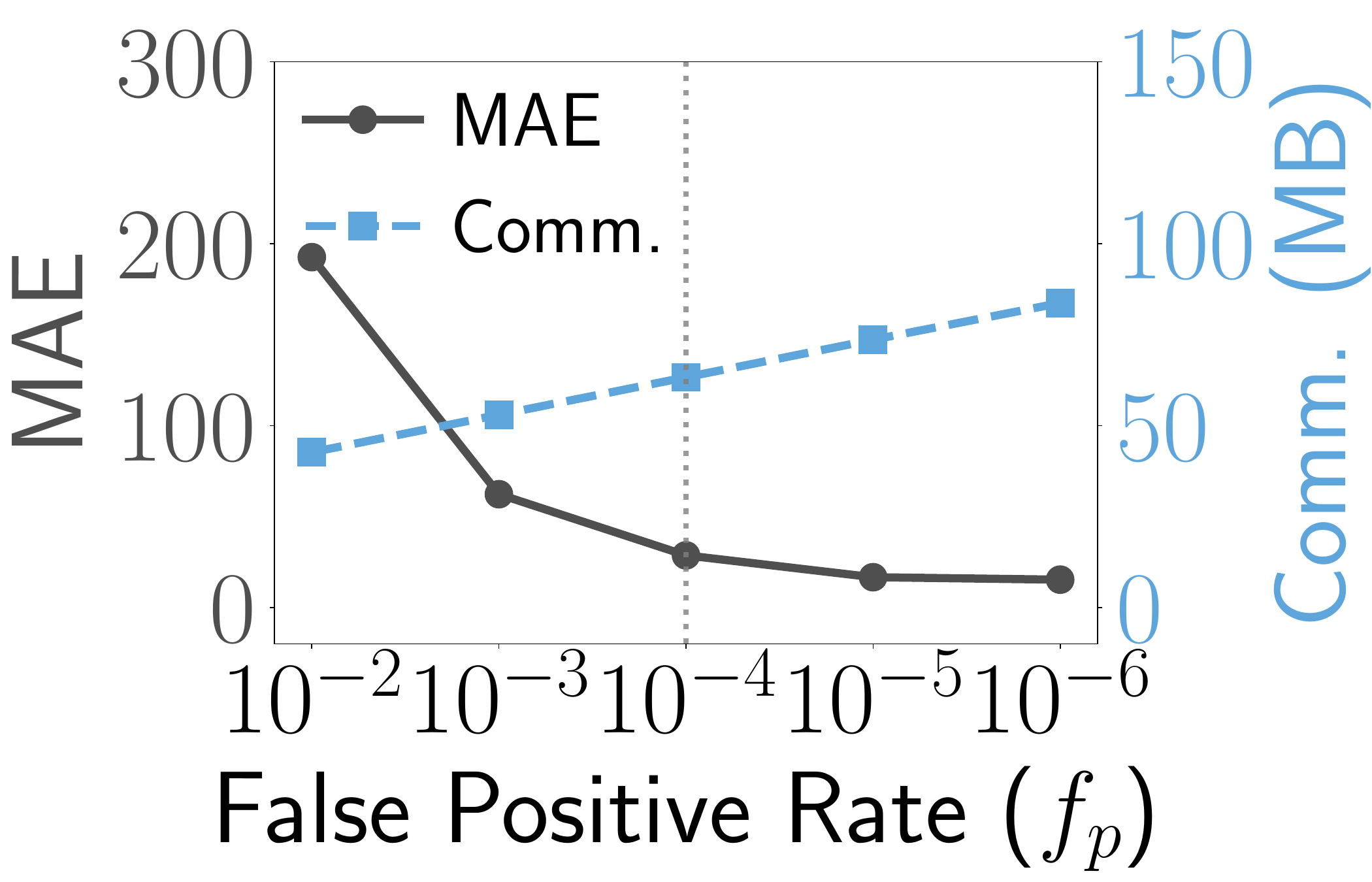}}
    \vspace{-0.25cm}
  \caption{The PPRC results for different query range sizes and Bloom filter's false positive rates $f_p$.}
  \label{fig:result3}
\end{figure}

\textit{DD-PGQ.} In practical scenarios, geographic queries often span private datasets held by multiple DHs who are unwilling to outsource data due to privacy concerns, rendering OD-PGQ approaches unsuitable.
To address this, DD-PGQ schemes have been proposed.
Shi et al. \cite{shi2021efficient} first investigate approximate range aggregation queries without privacy considerations, followed by Hu-fu \cite{tong2022hu}, the first DD-PGQ scheme based on cryptographic tools.
Zhang et al. \cite{zhang2023approximate} introduce efficient and secure algorithms for approximate k-nearest neighbor queries across distributed datasets.
RCC \cite{akhavan2023level} supports exact range queries while protecting the privacy of the query range. 
To alleviate this bottleneck, recent works \cite{li2023efficient, chen2025u} incorporate differential privacy (DP) \cite{dwork2006differential}. Their core idea is to release DP-perturbed values. 
Li et al. \cite{li2023efficient} introduce FedGroup, a private range counting scheme that combines cryptographic tools with DP to securely group DHs and adds DP noise at the group level.
Chen et al. \cite{chen2025u} propose U-DPAP, a DP-only range counting scheme. This scheme groups similar data and injects noise into the groups.
While offering improved efficiency, they are less accurate than cryptography-based solutions due to noise injection into the query result. Moreover, they provide weaker privacy guarantees, as the released perturbed data remains close to the original data.

Despite recent progress, existing DD-PGQ schemes either assume non-overlapping datasets or incur prohibitive overhead when protecting query privacy. 
Thus, achieving efficient, query-private, and deduplicated range counting over distributed geographic data remains an open challenge, which our work aims to address.

%% file: conclusions.tex
\section{Discussion}

\noindent \textbf{Functional extensions.}
Although PPRC is designed for rectangular range counting, it can be naturally extended to support arbitrary spatial regions (e.g., circular or polygonal areas), without modifying its core protocol. 
Specifically, Geohash-based encoding~\cite{guo2021luxgeo} and prefix compression~\cite{li2020adaptive} can represent target regions as compact prefix sets, which are directly compatible with the existing PRP and OLC design. 
Beyond the extension, another important future direction is to optimize PPRC for single-source range counting and single-point queries, which are naturally supported by our framework but are not the primary focus of this work.

\noindent \textbf{Supporting multiple QUs.}
Although PPRC is described under a single-QU setting, it can be extended to support multiple QUs because each query execution is independent. 
Concurrent requests can therefore be handled through the standard request-level parallelization at the DHs. 
To mitigate potential information leakage when the same QU submits multiple queries, PPRC can be complemented with standard orthogonal mechanisms such as query auditing, privacy-budget enforcement, rate limiting, or differential privacy.
Besides, batch query scheduling and preprocessing of query-independent data structures can further improve the overall system throughput without changing the core protocol design in high-throughput settings.

\noindent \textbf{Extension beyond geographic data.}
While motivated by geographic range counting, PPRC is not inherently restricted to geographic data. 
Its design only assumes that records can be represented as points in a multidimensional domain and that queries can be formulated as bounded multidimensional ranges.
Under this abstraction, PPRC can naturally support other range-counting scenarios, such as healthcare, sensor, and business analytics, without changing the core protocol.

\section{Conclusion} \label{sec:conclusions}
In this paper, we propose PPRC, the first protocol for Private Distributed Range Counting (PDRC) that simultaneously achieves accuracy across overlapping datasets, bilateral privacy, and practical efficiency.
PPRC adopts a secure sketching framework that integrates a secure range predicate (PRP) with an efficient aggregation scheme (OLC).
Specifically, PRP reformulates secure point-in-range evaluation as a secure membership test using encrypted Bloom filters, enabling efficient range evaluation without disclosing the query range. 
OLC securely aggregates overlapped partial results through lightweight cryptographic additions while ensuring no leakage beyond the final count.
We provide a formal analysis of the accuracy, efficiency, and security of PPRC.
Experiments on real-world and synthetic datasets demonstrate that PPRC achieves up to $55\times$ lower errors and $37\times$ faster performance compared to baseline protocols, establishing it as a practical solution to PDRC.